\documentclass[11pt]{article} 

\usepackage[margin=1in]{geometry}
\usepackage{graphicx}
\usepackage{dcolumn}
\usepackage{bm}
\usepackage{amsmath, amsthm, amssymb, amsfonts}
\usepackage{physics}
\usepackage{palatino}
\usepackage{mathtools}
\usepackage{algorithm}
\usepackage{algpseudocode}
\usepackage{dsfont}
\usepackage{xcolor}
\usepackage{comment}
\usepackage{subcaption}
\usepackage{rotating}
\usepackage[normalem]{ulem}
\usepackage{tikz}
\usepackage{siunitx}
\usepackage[outline]{contour} 
\usepackage{url}

\usetikzlibrary{quantikz}
\usetikzlibrary{angles,quotes} 
\usetikzlibrary{backgrounds,fit}
\contourlength{1.3pt}
\tikzset{>=latex} 
\tikzstyle{vector}=[->,very thick]
\tikzstyle{mydashed}=[dash pattern=on 2pt off 2pt]
\tikzstyle{startstop} = [rectangle, rounded corners, minimum width=2.5cm, minimum height=1cm, text centered, draw=black, fill=red!30, text width=2cm]
\tikzstyle{process} = [rectangle, minimum width=2.5cm, minimum height=1cm, text centered, draw=black, fill=blue!30, text width=2cm]
\tikzstyle{lastnode} = [rectangle, minimum width=5cm, minimum height=1cm, text centered, draw=black, fill=blue!30]
\tikzstyle{arrow} = [thick,->,>=stealth]

\newtheorem{theorem}{Theorem}[section] 

\newtheorem{lemma}[theorem]{Lemma}
\newtheorem{strategy}[theorem]{Strategy}
\newtheorem{remark}[theorem]{Remark}

\newcommand{\Ocomp}{\mathcal{O}}

\newcommand{\Iset}{\mathcal{I}}

\newcommand{\hsip}[2]{\left\langle #1, #2 \right\rangle}

\newcommand{\rvline}{\hspace*{-\arraycolsep}\vline\hspace*{-\arraycolsep}}

\newcommand{\st}{\mathrm{s.t.}}

\newcommand{\id}{\mathds{1}}
\newcommand{\reals}{\mathbb{R}}

\newcommand{\kb}[1]{ \ket{#1} \!\! \bra{#1}}

\begin{document}

\title{A hybrid quantum-classical algorithm for Bayes-optimal quantum state discrimination using the source code}

\author{
Ankith Mohan\thanks{Department of Computer Science, Virginia Polytechnic Institute and State University, Blacksburg, VA, USA $24061$. \texttt{ankithmo@vt.edu}}
\and
Jamie Sikora\thanks{Department of Computer Science, Virginia Polytechnic Institute and State University, Blacksburg, VA, USA $24061$. \texttt{sikora@vt.edu}}
\and
Sarvagya Upadhyay\thanks{Fujitsu Research of America, Sunnyvale, CA, USA $94085$.
\texttt{supadhyay@fujitsu.com}}
}

\date{\today} 

\maketitle

\begin{abstract}
Quantum state discrimination is a fundamental primitive in quantum information processing, underpinning tasks in quantum communication, sensing, and learning. 
We consider the general Bayes framework, as introduced by Helstrom, for state discrimination when, instead of a classical description of the candidate states, one has access to their \emph{source code}: the quantum circuit that prepares them.
We show that the semidefinite program (SDP) for the discrimination problem can be reformulated in terms of the Gram matrix of these states, reducing the SDP variable dimensions from $dL$ to $NL$, where $d$ is the Hilbert space dimension,
$N$ is the number of candidate states, and $L$ is the number of possible guesses.
Importantly, we further introduce a quantum pre-processing procedure which 
efficiently constructs the reduced semidefinite program from the source code, enabling our method to operate directly on quantum data. 
We consider two applications.
First, we characterize the optimal identifications for quantum changepoint problems under several reward structures, including multiple-changepoint settings that were previously computationally inaccessible.
Second, we consider a quantum error classification problem and show how our reduction makes it tractable for systems of hundreds of qubits.
\end{abstract}

\section{Introduction}\label{sec:intro} 

Quantum state discrimination is a central task in quantum information processing.
Quantum key distribution~\cite{ko2018advanced}, two-party cryptography~\cite{aharonov2000quantum, ambainis2001new, sikora2017simple}, the hidden subgroup problem~\cite{hayashi2008quantum}, and dimension witnessing~\cite{strubi2013measuring}, among many others, can all be formulated in terms of this problem, which can be described as the following two-party game.
Alice chooses a state from a fixed set of states and sends it to Bob who wishes to identify the state. 
Here we assume that both Alice and Bob know the set of states and the a priori probabilities before the game commences.

For Bob to determine his state perfectly, the necessary and sufficient condition is that the states in the set must be pairwise orthogonal.
However this need not be the case and, in such circumstances, Bob may wish to optimize some desired figure of merit, e.g., the average probability of correctly identifying the state.
Several other figures of merit exist depending on the desired application.
For example, Bob may wish to only guess when he knows he will be correct, called unambiguous state discrimination.
Each figure of merit could correspond to adopting a different strategy which would in turn lead to different behaviour.  
We discuss several variants in this work and refer the interested reader to the following reviews~\cite{chefles2000quantum, bergou200411, bergou2007quantum, barnett2009quantum, bergou2010discrimination}. 

\subsection{Quantum state discrimination}

We now discuss the quantum state discrimination problem.
Formally, suppose Alice selects the state $\rho_i$ from the set of states $\{\rho_1, \dots, \rho_N\}$ with corresponding a priori probabilities $\{q_1, \dots, q_N\}$.
For Bob to make a guess, he must construct POVM operators $\{M_1, \dots, M_L\}$, such that a detection event on the $i$-th operator, $M_i$, corresponds to Bob guessing $i$. 
Note that the number of these operators, $L$, depends on the figure of merit adopted. 
By Born's rule, the probability of a detection event on the $i$-th operator upon receiving the state $\rho_j$ is given by
\begin{align}
    \label{eq:Jason_Bourne} 
    \Pr[i | j] = \hsip{M_i}{\rho_j}. 
\end{align} 

We now discuss several variants, each of which correspond to how we want $i$ and $j$ to be related. 

\paragraph{Minimum error discrimination.}
In this variant, Bob wishes to guess which state was sent. 
That is, upon observing a detection outcome on the $i$-th operator, Bob concludes that he received the state $\rho_i$.
Bob's goal is to maximize the correlation between $i$ and $j$ in this case. 
The optimal guessing probability can be computed as the optimal value of the optimization problem below 
\begin{align}
    \label{eq:ME}
    \text{maximize:} \; \left\{ \sum_{i=1}^N q_i \hsip{M_i}{\rho_i} \right\}.
\end{align}

Distinguishing between two states is well-studied and the solution is given by the \emph{Helstrom measurement}~\cite{helstrom1969quantum}.  
When dealing with more than two states, a solution is known for geometrically uniform states \cite{eldar2004optimal}, for mirror-symmetric states~\cite{andersson2002minimum}, and for symmetric states generated by a specific class of unitary matrices~\cite{barnett2001minimum, chiribella2004covariant, chiribella2006maximum}. 
Also, \cite{bae2013structure} proposes a general approach by exploiting the geometry of the states.
However, for most cases, a closed-form solution is unknown and one must resort to numerical solvers~\cite{cvx} to find the optimal measurements.

\paragraph{Unambiguous discrimination.} 
This variant is similar to the minimum error case, except we further restrict Bob to \emph{never make a mistake}. 
In other words, Bob only makes a guess when he is certain about his received state, and in all other cases he returns an inconclusive outcome corresponding to ``I do not know!''. 
This can be achieved through constraining the POVM $\{M_1, \dots, M_N,M_{N+1}\}$ to satisfy $\Pr[i|j] = \hsip{M_i}{\rho_j} = 0$, for all $i \ne j$. 
Here, we have $L = N+1$ with $M_{N+1}$ corresponding to the inconclusive outcome.
The guessing probability in this setting is equal to the optimal value of the following optimization problem 
\begin{align}
    \label{eq:UA}
    \text{maximize:} \;  \left\{ \sum_{i=1}^N q_i \hsip{M_i}{\rho_i}  : \hsip{M_i}{\rho_j} = 0, \text{ for all }\ i \ne j \right\}.
\end{align}

The unambiguous case is often restricted to the cases where the states are either linearly independent for pure states~\cite{chefles1998unambiguous}, or when their supports do not completely overlap for mixed states~\cite{rudolph2003unambiguous}. 
Similar to the minimum error case, discriminating between two states unambiguously has a closed form solution given by Jaeger and Shimony~\cite{jaeger1995optimal}. 

\paragraph{Generalizations of these strategies.} 
From the discussion so far, we can express the probability of correctly identifying the states, the error probability, and the probability of obtaining an inconclusive outcome below 
\begin{align}
    \label{eq:P}
    P_D \coloneqq \sum_{i=1}^N q_i\, \hsip{M_i}{\rho_i}, \quad P_E \coloneqq \sum_{i,j: i \neq j}^{N} q_j\, \hsip{M_i}{\rho_j}, \quad \text{ and } \quad  P_I  \coloneqq \sum_{j=1}^N q_j\, \hsip{M_{N+1}}{\rho_j},
\end{align}
respectively. 
In unambiguous discrimination, the operators must minimize $P_I$, equivalently maximize $P_D$, subject to $P_E = 0$.
In these cases, it may be prudent to relax this constraint by setting $P_E \le \epsilon$ for some $\epsilon \in (0, 1)$.
Introduction of such error is useful in quantum key distribution, for example, where the amount of information that can be learned by an eavesdropper can be decreased through this error~\cite{tamaki2003security}.
Some constraints were established in~\cite{touzel2007optimal} on the optimal solution for the unambiguous case with fixed error probability.
A similar idea for minimum-error discrimination was investigated in~\cite{eldar2003mixed}. 

\paragraph{Quantum state exclusion.} 
Consider the same game where Alice picks a state and sends it to Bob, but instead of Bob identifying what state he received, here the task is to \emph{exclude} a state he did \emph{not} receive.
For instance, if Bob receives the state $\rho_1$ and he replies ``I did not receive $\rho_2$'', this corresponds to a winning scenario.
Here, Bob wishes to construct a POVM $\{M_1, \dots, M_N\}$ that minimizes the detection outcome $M_j$ upon receiving the state $\rho_j$. 
The error probability for minimum-error state exclusion can be computed as the optimal value of the following optimization problem 
\begin{align}
    \label{eq:ME_SE}
    \text{minimize:} \; \left\{ \sum_{i=1}^N q_i \hsip{M_i}{\rho_i} \right\}. 
\end{align}
This problem has been studied in detail in~\cite{bandyopadhyay2014conclusive, uola2020all, ducuara2019operational}.
A set of quantum states is said to be \emph{antidistinguishable} if there exists a POVM that can perfectly exclude the state Bob did not receive, i.e., the error above is $0$. 
Necessary and sufficient conditions when a set is antidistinguishable are derived in~\cite{russo2023inner},  tight bounds are given in~\cite{johnston2023tight}, and optimal error exponents are discussed in~\cite{mishra2023optimal}. 
 
For the unambiguous variant of quantum state exclusion, Bob now wishes to construct a POVM $\{M_1, \dots, M_N, M_{N+1}\}$ such that detection outcome $M_j$ excludes with certainty the state $\rho_j$ and the detection on $M_{N+1}$ leads to a situation where Bob does not exclude any state. 
The goal is to minimize the probability of obtaining the inconclusive outcome, while ensuring that a detection on $M_j$ does not exclude the state $\rho_j$.
The optimal error probability of this scenario is given by the optimal value of the following problem 
\begin{align}
    \label{eq:UA_SE}
    \text{minimize:} \; \left\{ \sum_{j=1}^N q_j \hsip{M_{N+1}}{\rho_j} : \hsip{M_j}{\rho_j} = 0 \text{ for all } j \right\}. 
\end{align}

We now discuss a general framework to study quantum state discrimination problems.

\subsubsection{The general Bayes approach}\label{sec:gen_frame} 

Recall the setting that Bob expects one of the states from the set $\{\rho_1, \dots, \rho_N\}$ with corresponding a priori probabilities $\{q_1, \dots, q_N\}$. 
Helstrom's Bayes approach~\cite{helstrom1969quantum} considers a POVM $\{M_1, \dots, M_{L}\}$ along with the parameters $R_{ij}$, where $R_{ij}$ is Bob's reward\footnote{The reward can also be negative which can be interpreted as a penalty.} if he outputs $i$, given by $M_i$, upon receiving the state $\rho_j$. 
The idea behind the reward is that it generalizes the aforementioned discrimination tasks while opening up the possibilities for new ones. 
The parameter $R_{ij}$ is also referred to as the cost matrix~\cite{helstrom1969quantum, holevo1973statistical, holevo1978asymptotically, brody1996bayesian, helstrom2003bayes, yuen2003optimum, nakahira2012minimum}, loss function~\cite{wieczorek2018entropic} or the payoff function~\cite{watrous2018theory}.

In general, we can assume that Bob wishes to maximize his reward function given as 
\begin{align}
    \label{eq:gen_QSD}
    \text{maximize:} \;  \left\{ \sum_{i=1}^L \sum_{j=1}^N R_{ij} q_j  \hsip{M_i}{\rho_j} \right\}.
\end{align} 

We describe below how some discrimination tasks fit into this framework.

\begin{itemize} 
\item \textbf{Minimum-error state discrimination}: This is the case when $L = N$, and
\begin{align}
    \label{eq:R_ME_SD}
    R_{ij} =
    \begin{cases}
        1 &\text{if} \quad i = j \\
        0 &\text{otherwise}
    \end{cases}.
\end{align}
Notice here that Bob is given a unit reward only when he is able to correctly identify the state he received, and gets no reward when he makes an error.
Applying this reward scheme to Eq.~\eqref{eq:gen_QSD}, we can recover Eq.~\eqref{eq:ME}. 

\item \textbf{Minimum-error state exclusion}: 
This is the case when $L = N$, and
\begin{align}
    \label{eq:R_ME_SE}
    R_{ij} =
    \begin{cases}
        0 &\text{if} \quad i = j \\
        1 &\text{otherwise}
    \end{cases}. 
\end{align}
Substituting this reward scheme in Eq.~\eqref{eq:gen_QSD} recovers Eq.~\eqref{eq:ME_SE}. 

\item \textbf{Unambiguous state discrimination}: 
This is the case when $L = N+1$, with $i = N+1$ indicating an inconclusive measurement, and
\begin{align}
    \label{eq:R_UA_SD}
    R_{ij} =
    \begin{cases}
        1 &\text{if} \quad i = j \\
        -\infty &\text{if} \quad i \ne j \\
        0 & \text{if} \quad i = N+1 
    \end{cases}. 
\end{align} 
Bob is awarded a unit reward when he correctly identifies his received state and is granted a penalty of $-\infty$ for all erroneous situations to ensure that such errors are not committed. 
Also, he receives no reward (or penalty) for the inconclusive outcome $N+1$.  
Using this reward scheme in Eq.~\eqref{eq:gen_QSD} gives us Eq.~\eqref{eq:UA}. 

A variant of this scheme is 
\begin{align}
    \label{eq:R_UA_SD}
    R_{ij} =
    \begin{cases} 
        1 &\text{if} \quad i = j \\
        -\beta &\text{if} \quad i \ne j \\
        0 & \text{if} \quad i = N+1 
    \end{cases} 
\end{align} 
for some $\beta > 0$. 
This allows a penalty for an incorrect guess, but is not completely forbidden. 

\item \textbf{Unambiguous state exclusion}:  
This is the case when 
\begin{align}
    \label{eq:R_UA_SE}
    R_{ij} =
    \begin{cases}
        -\infty &\text{if} \quad i = j \\
        0 &\text{if} \quad i \ne j \\
        -1 & \text{if} \quad i = N+1 
    \end{cases}. 
\end{align} 
Here Bob is not allowed to make any errors, therefore he receives a penalty of $-\infty$ whenever he excludes his received state.
\end{itemize}  

Other reward functions that are relevant to a wide range of quantum problems can be considered through this approach as well.
In many settings, an incorrect guess may still be preferred over other incorrect guesses. An important example of this perspective arises in anomaly detection~\cite{llorens2024quantum} and changepoint detection~\cite{sentis2016quantum, sentis2017exact}. 

\begin{itemize} 

\item \textbf{The horseshoe reward}: This is the same as the minimum error case, but one does not need to be exact to be correct, \emph{just close enough}.  
By setting $\mu$ to be the \emph{closeness parameter}, a non-negative integer, and $L = N$, define  
\begin{align}
    \label{eq:R_ME_horseshoe}
    R_{ij} =
    \begin{cases}
        1 &\text{if} \quad |i - j| \leq \mu \\
        0 &\text{otherwise}
    \end{cases}. 
\end{align} 
This gives a unit reward if one is close enough, and nothing otherwise.  
By setting $\mu = 0$, we recover minimum error discrimination, Eq.~(\ref{eq:ME}). 

\item \textbf{The closer-the-better reward}: This is when one does not need to be exact to be correct, \emph{but closer is better}.  
By setting $\gamma \in [0,1]$ to be some fixed reward, and $L = N$, define  
\begin{align}
    \label{eq:R_ME_ctbr}
    R_{ij} = \gamma^{|i-j|}. 
\end{align} 
This gives a greater reward the closer one is to correctly identifying the state. 
By setting $\gamma = 0$, we recover minimum error discrimination, Eq.~(\ref{eq:ME}). 
By setting $\gamma = 1$, we get the every-guess-is-a-winner reward. 

\item \textbf{The high school exam reward}: This is when one is rewarded partial marks for saying ``I do not know''.  
Here $L = N+1$, and
\begin{align}
    \label{eq:R_ME_school}
    R_{ij} =
    \begin{cases}
        1 &\text{if} \quad i = j \\ 
        0 &\text{if} \quad i \neq j \\
        0.25 & \text{if} \quad i = N+1 
    \end{cases}. 
\end{align} 
This deters from random guessing and awards partial credit to an inconclusive guess.  
This variant is studied in \cite{combes2015cost} which they call the $0-1-\lambda$ reward where $1-\lambda$ is the reward that Bob receives for saying ``I do not know''. They examine how the measurement changes as the value of $\lambda$ is varied.

\end{itemize} 

\subsubsection{Classification and mixed states}

This approach also allows one to capture a large class of discrimination problems known as \emph{classification}.
In this task, each state belongs to a subset of the classes $\{C_1, \dots, C_L\}$, and Bob does not need to identify exactly which state was given, he just needs to identify a class that it belongs to. 
For example, if Bob is given a picture of an animal, a sufficient guess could be ``cat'' as opposed to identifying exactly which cat.

For this task, we can define 
\begin{align}
    \label{eq:class}
    R_{ij} =
    \begin{cases}
        1 &\text{if} \quad \rho_j \in C_i \\ 
        0 &\text{otherwise}
    \end{cases}. 
\end{align}
(Here, $i$ runs over the number of classes.) 
This means that Bob gets a unit reward if he correctly guesses a class that $\rho_j$ belongs to, i.e., he correctly classifies the given state. 

\paragraph{Mixed states.}
While many state discrimination results apply to the pure state case, we now show how this framework can be applied to mixed states as well.
Suppose that the spectral decomposition of the state $\rho_j$ is given by 
\begin{equation}
    \rho_j = \sum_{k_j=1}^{r_j} \lambda_{jk_j} \ketbra{\psi_{jk_j}}{\psi_{jk_j}},
\end{equation}
where $r_j$ is the rank of the state $\rho_j$.
Then the reward function in Eq.~\eqref{eq:gen_QSD} can be expressed as
\begin{align}
    \label{eq:mixed_states}
    \sum_{i=1}^L \sum_{j=1}^N R_{ij} q_j \hsip{M_i}{\rho_j} 
    &= \sum_{i=1}^L \sum_{j=1}^N R_{ij} q_j \hsip{M_i}{\sum_{k_j=1}^{r_j} \lambda_{jk_j} \ketbra{\psi_{jk_j}}} \nonumber \\
    &= \sum_{i=1}^L \sum_{j=1}^N \sum_{k_j=1}^{r_j} R_{ij} q_j \lambda_{jk_j} \hsip{M_i}{\ketbra{\psi_{jk_j}}} \nonumber \\
    &= \sum_{i=1}^L \sum_{j=1}^N \sum_{k_j=1}^{r_j} R_{ij} \tilde{q}_{jk_j} \hsip{M_i}{\ketbra{\psi_{jk_j}}}
\end{align}
where $\tilde{q}_{jk_j} = q_j \lambda_{jk_j}$ denotes the apriori probability of Bob receiving the state $\ket{\psi_{jk_j}}$. 

Note that the number of pure states is no longer $N$, but rather $r$ where $r = \sum_{j=1}^N r_{j}$ is the sum of the ranks. 
Therefore, if all of the states have low rank then we still have a small (pure state) instance of the problem. 
Thus, by defining the rewards $R_{ijk_j} = R_{ij}$, we can model the general setting for mixed states using the pure state version.

As a concrete example, the case of minimum error discrimination of mixed states can be defined using the rewards 
\begin{align}
    \label{eq:min_class}
    R_{ijk_j} =
    \begin{cases}
        1 & \text{if} \quad i = j \\ 
        0 & \text{if} \quad i \neq j 
    \end{cases}. 
\end{align}

\medskip 
\begin{remark}
Although our results are presented for pure states, the preceding discussion illustrates how this framework can handle the mixed state case as well.  
\end{remark}

\begin{remark}
We emphasize that the advantage of our reduction is most pronounced when the states are pure or low-rank. 
When the ranks are large, the effective number of pure states can become large, and the computational savings diminish. 
\end{remark}
 
\section{Technical results}

In this section, we describe our semidefinite programming reduction and the hybrid algorithm for computing the optimal reward.
In Section~\ref{sec:apps} we show applications and numerical simulations.

\subsection{Semidefinite programming formulations and reductions} 

Given a set of rewards $R_{ij}$, we can model the optimization of the optimal reward as the optimal objective function value of the following semidefinite program.
\begin{align}
    \label{eq:gen_QSD2}
    \alpha = \text{maximize:} \;  \left\{ \sum_{i=1}^L \sum_{j=1}^N R_{ij} q_j  \hsip{M_i}{\ketbra{\psi_j}} \right\}   
\end{align} 
We now investigate how hard it is to compute $\alpha$ and to find optimal POVMs.  

We remark that the quantity given in (\ref{eq:gen_QSD2}) is the optimization of a linear function over affine constraints $\sum_{i=1}^L M_i = \id$ and each variable is positive semidefinite. 
Thus, it can be written as a semidefinite program (SDP) (see Appendix~\ref{sec:bg}).
However, this computation involves finding $L$ operators each of size $d \times d$ which is the size of each of the states. 
Thus, if the states involve many qubits, this is an intractable problem. 

To reduce the size of the SDP problem, we investigate the following equivalent SDP, below  
\begin{align} 
    \alpha' & = \text{maximize:} \;
        \left\{ 
            \sum_{i=1}^L \sum_{j=1}^N R_{ij} q_j \bra{j} W_i \ket{j} : 
            \sum_{i=1}^L W_i = G,\ W_i \succeq 0 
        \right\} \label{reducedprimal}
\end{align}
where $G$ is the Gram matrix of the states $\{\ket{\psi_i}\}$, i.e., 
\begin{equation} 
G = \sum_{i,j=1}^N \langle \psi_i | \psi_j \rangle \ketbra{i}{j}. 
\end{equation} 
We show the equivalence between the two SDPs in Section~\ref{sec:SDPs}. 
This SDP has the advantage of (typically) being much smaller in size compared to the one given by~(\ref{eq:gen_QSD2}) and in general much easier to solve. 
Indeed, they involve $L$ SDP variables each of size $N$ now. 
Their utility is given by the following theorem. 

\begin{theorem} \label{thm:reduction} 
Given $N$ fixed pure states $\{ \ket{\psi_1}, \ldots, \ket{\psi_N} \}$ and a priori probabilities $q_1, \ldots, q_N$, the general Bayes approach given in \eqref{eq:gen_QSD2} can be calculated by \eqref{reducedprimal}. 
More precisely, $\alpha = \alpha'$.
\end{theorem}

While restricting to the support subspace is well known\footnote{The minimal reducing subspace technique~\cite[Proposition 1]{holevo1982testing} is an example of such an approach.}, this reduction is the explicit reformulation as an SDP over $N \times N$ matrices parameterized solely by the Gram matrix $G$. 
This reformulation:
(i) decouples the SDP from any Hilbert-space representation, 
(ii) enables direct construction from experimentally accessible inner products via the hybrid algorithm of Strategy 1.3, and 
(iii) yields a concrete computational pipeline that can scale with the number of states in the problem, which is illustrated in Section~\ref{sec:apps}.

\paragraph{Related works that focus on Gram matrices.}
Several works exploit the fact that the discrimination properties for a set of states are captured by the Gram matrix \emph{for fixed applications}. 
For instance,~\cite{sentis2016quantum, sentis2017exact, sentis2018online} use this to prove bounds on the success probability for the changepoint problem,~\cite{llorens2024quantum} draws on this for the multi-anomaly detection problem, and \cite{skotiniotis2024identification} utilizes this for error position identification. 
In~\cite{dalla2015optimality}, the authors characterize the optimal success probability of discriminating a set of linearly independent pure states in the minimum error regime as a function of the Gram matrix.
For discriminating between linearly independent pure states with an intrinsic ordering,~\cite{martinez2019certified} obtains a bound for the success probability when using a combination of the horseshoe reward (Eq.~\eqref{eq:R_ME_horseshoe}) and unambiguous discrimination. 
In \cite{johnston2023tight, johnston2025complexityperfectquantumstate}, analyzing the task of \emph{$k$-learnability} makes significant use of a reduction of the problem to properties of the Gram matrix. 
These works demonstrate that not only is the reduction to a Gram matrix attractive from a computational perspective, but also from an analytical perspective. 
Our reduction in this work takes the above approaches and unifies them for general, application-agnostic state discrimination tasks. 
Perhaps surprisingly, this has not been yet been considered and we suspect this will be helpful for many other future applications.

\subsection{A hybrid algorithm for calculating the optimal reward}

As mentioned above, if we had access to the inner products, we would be able to solve for the optimal reward function values (for any reward of our choosing).
However, computing inner products of exponentially large vectors is expensive. 
Therefore, we now discuss means of doing this on a quantum computer. 
After all, the quantities involve a property of physical quantum states. 

Since the states $\{\ket{\psi_i}\}$ are prepared and sent to Bob, we make the assumption that they are efficiently preparable. 
Given this, Bob can perform some pre-computations on them before the discrimination game starts. 
For instance, Bob can perform the Hadamard test to learn Re($\braket{\psi_i}{\psi_j}$) and/or Im($\braket{\psi_i}{\psi_j}$) for all values of $i$ and $j$ (see Section~\ref{sec:QC_bg}), assuming he has access to a unitary $U_{ij}$ that can map the state $\ket{\psi_i}$ to $\ket{\psi_j}$ for all values of $i$ and $j$. 

We also note that if Bob is given access to unitaries that prepare each of the states $\ket{\psi_i}$, then a block-encoding of the Gram matrix can be efficiently implemented (see Lemma 47 of \cite{gilyen2019quantum}). 

For completeness, we mention that in the case where it is hard to implement the controlled unitary required for the Hadamard test, if the unitary $U$ can be decomposed into either the
(1) sum of Pauli products with a number of terms that is polynomial in the number of qubits $n$, or
(2) tensor product of unitaries $U_q$ where each $U_q$ acts on at most $\Ocomp(\text{poly}(\log n))$ qubits,
then the term $\braket{\psi_i}{\psi_j}$ can be estimated using a direct measurement method~\cite{mitarai2019methodology}.
Moreover, if each of the states $\{\ket{\psi_i}\}$ can be prepared as $P_i \ket{\psi}$ using an initial state $\ket{\psi}$ and Pauli string $P_i \in \{\id, X, Y, Z\}^{\otimes t}$ for some $t$, then note that $\braket{\psi_i}{\psi_j} = \bra{\psi} P_i P_j \ket{\psi} = a \bra{\psi} P' \ket{\psi}$ where $P'$ is some Pauli string and $a \in \{+1, -1, +i, -i\}$.
Then calculating the inner product amounts to computing the expectation value of the Pauli string $P'$~\cite{bharti2021iterative}.

In the special case where the inner products are all non-negative, we can use the standard swap test to compute these values. 
Given two states $\ket{\psi}$ and $\ket{\phi}$, the swap test has two outcomes, one occurring with probability $\frac{1}{2} + \frac{1}{2} | \braket{\psi}{\phi} |^2$ and the other occurring with probability $\frac{1}{2} - \frac{1}{2} | \braket{\psi}{\phi} |^2$ (we discuss this more in Appendix~\ref{sec:QC_bg}).
  
Figure~\ref{fig:flow_chart} depicts this hybrid algorithm.
\begin{figure}[ht]
    \centering
    \begin{tikzpicture}[node distance=3cm]
        \node (U) [startstop] {Preparation unitaries};
        \node (ip) [process, right of=U] {Compute inner products};
        \node (G) [startstop, right of=ip] {Gram matrix $G$};
        \node (SDP) [process, below of=ip] {Reduced SDP};
        \node (R) [startstop, right of=SDP] {Reward matrix $R$};
        \node (q) [startstop, left of=SDP] {A priori probability vector $q$};
        \node (p) [lastnode, below of=SDP] {Optimal reward};

        \begin{scope}[on background layer]
            \node[fill=blue!10, draw=blue!30,
                  fit=(U)(ip)(G), label={[blue!70, font=\bfseries]above:Quantum}] {};
            \node[fill=red!10, draw=red!30,
                  fit=(q)(SDP)(R), label={[red!70!black, font=\bfseries]above:Classical}] {};
        \end{scope}
        
        \draw [arrow] (U) -- (ip);
        \draw [arrow] (ip) -- (G);
        \draw [arrow] (G) -- (SDP);
        \draw [arrow] (R) -- (SDP);
        \draw [arrow] (q) -- (SDP);
        \draw [arrow] (SDP) -- (p);
    \end{tikzpicture}
    \caption{A visual representation of the hybrid algorithm.}
    \label{fig:flow_chart}
\end{figure}

\begin{strategy} 
Suppose we are given $N$ pure states $\{ \ket{\psi_1}, \ldots, \ket{\psi_N} \}$ and a priori probabilities $\{q_1, \ldots, q_N\}$, 
and suppose further that we are given efficient preparation unitaries that prepare each of the states. 
Then we can approximate the optimal reward function using a hybrid quantum-classical algorithm via the Hadamard test. 
If the quantum states are not known but the inner products are promised to be non-negative, the same holds only given access to preparation devices via the swap test.
In each case, the classical part of the algorithm involves solving SDPs with $L$ variables each of size $N \times N$.
\end{strategy}

We remark that in this work we do not examine estimation errors (say, in computing the inner products) and how sensitive the SDPs are to these errors. 
As far as we are aware, general sensitivity analyses are not known for SDP algorithms. 
As such, we leave this as an interesting open problem and rely simply on using machine-accurate estimations of such parameters.

\subsubsection{Tests in quantum computing}\label{sec:QC_bg} 

\paragraph{Swap test.}

Given two (possibly unknown) pure states $\ket{\psi}$ and $\ket{\phi}$, the swap test outputs a Bernoulli random variable that is $0$ with probability $\frac{1}{2} + \frac{1}{2}|\braket{\psi}{\phi}|^2$.
Figure~\ref{fig:inner_product_circuit} (a) depicts the corresponding circuit.
One can learn $|\braket{\psi}{\phi}|^2$ to additive accuracy $\epsilon$ with failure probability at most $\delta$ using $\Ocomp \left(\frac{1}{\epsilon^2} \log \left( \frac{1}{\delta} \right) \right)$ copies of the states and $\tilde{\Ocomp} \left(\frac{1}{\epsilon^2} \log \left( \frac{1}{\delta} \right) \right)$ operations~\cite{buhrman2001quantum,huang2019near}.

\paragraph{Hadamard test.}

Given access to a unitary $U$ that maps the state $\ket{\psi}$ to the state $\ket{\phi}$, this method creates a random variable whose expected value is the real part of $\braket{\psi}{\phi}$.
Figure~\ref{fig:inner_product_circuit} (b) illustrates the circuit for this procedure.
The imaginary part of $\braket{\psi}{\phi}$ can be computed by applying a phase gate $S$ on the first qubit before the controlled unitary operation.
We can learn $|\braket{\psi}{\phi}|^2$ to additive accuracy $\epsilon$ with failure probability at most $\delta$ using $\Ocomp \left(\frac{1}{\epsilon^2} \log \left( \frac{1}{\delta} \right) \right)$ copies of the states and $\tilde{\Ocomp} \left(\frac{1}{\epsilon^2} \log \left( \frac{1}{\delta} \right) \right)$ operations~\cite{aharonov2006polynomial,huang2019near}.

\bigskip 

\begin{figure}[ht]
    \centering
    \begin{subfigure}[b]{0.495\textwidth}
        \centering
        \begin{quantikz}
            \lstick{$\ket{0}$} & \gate{H} & \ctrl{2} & \gate{H} & \meter{} & \cw \\
            \lstick{$\ket{\psi}$} & \qw & \swap{1} & \qw & \qw & \qw \\
            \lstick{$\ket{\phi}$} & \qw & \swap{} & \qw & \qw & \qw
        \end{quantikz}    
        \caption{Circuit for swap test.}
    \end{subfigure}
    \hfill
    \begin{subfigure}[b]{0.495\textwidth}
        \centering
        \begin{quantikz}
            \lstick{$\ket{0}$} & \gate{H} & \ctrl{1} & \gate{S^b} & \gate{H} & \meter{} & \cw \\
            \lstick{$\ket{\psi}$} & \qw & \gate{U} & \qw & \qw & \qw & \qw \\
        \end{quantikz}    
        \caption{Circuit for Hadamard test.}
    \end{subfigure}
    \caption{
    Circuits to learn the inner product between two (perhaps unknown) states $\ket{\psi}$ and $\ket{\phi}$.
    The swap test can approximate $|\braket{\psi}{\phi}|$ from many measured samples. 
    When $b = 0$, the Hadamard test approximates from many samples $\Re(\braket{\psi}{\phi})$, and when $b = 1$, $\Im(\braket{\psi}{\phi})$ can be approximated instead.
    The unitary $U$ maps the state $\ket{\psi}$ to the state $\ket{\phi}$. 
    }
    \label{fig:inner_product_circuit} 
\end{figure}

\section{Applications}\label{sec:apps}

We now discuss two problems that can be formulated as instances of the general Bayes approach and then apply our reduced SDPs followed by a discussion of numerical performance.

\subsection{The quantum changepoint identification problem.} 

Suppose that Alice sends to Bob a sequence of states $\ket{\psi}^{\otimes c_1} \otimes \ket{\phi_1}^{\otimes (c_2 - c_1)} \otimes \ket{\phi_2}^{\otimes (N - c_2)}$, one state at a time. 
We can think of $\ket{\psi}$ as the original state and $\ket{\phi_1}$ and $\ket{\phi_2}$ as mutated states. 
Perhaps Alice bought a cheap QKD device which started breaking down at time steps $c_1$ and $(c_2 - c_1)$. 
In this problem, Bob wishes to guess the index of the time steps, see Figure~\ref{fig:QCP} for an illustration. 
This is a quantum version of a \emph{changepoint problem} which is a widely studied field in statistical analysis with diverse applications, see the survey~\cite{aminikhanghahi2017survey}. 

\begin{figure}[h]
    \centering
    \includegraphics[width=\textwidth]{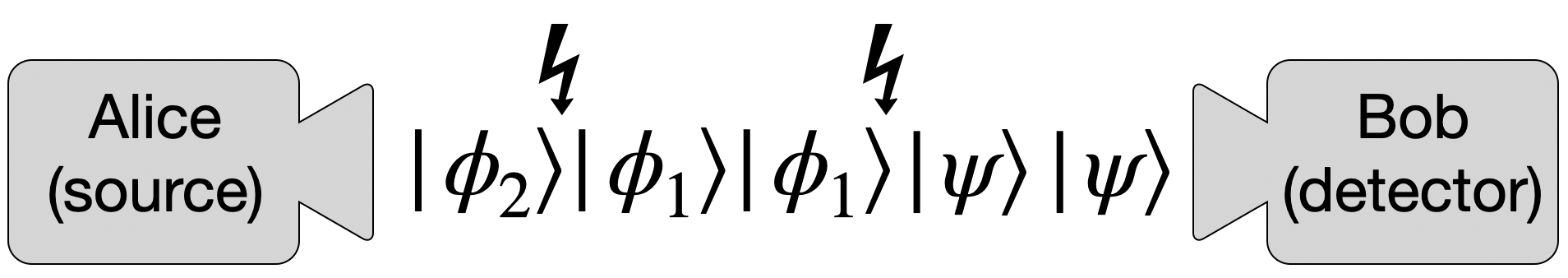}
    \caption{
    A quantum source is sending a stream of $\ket{\psi}$ states but two mutations occurred in this example. 
    The task is to design an optimal detector to guess when the mutations occurred. 
    }
    \label{fig:QCP}
\end{figure} 

Here we assume that both the promised state as well as the mutated ones are known to Bob. 
(We also discuss the case when they are unknown but we have a source of them to perform some pre-computations.) 
Observe that this problem can be formulated as an instance of the state discrimination problem. 
Consider the set of states consisting of every possible sequence. 
Then Bob's task is to construct a POVM that can identify the changepoints, where each possible guess  corresponds to a particular state in our set. 

The changepoint problem \emph{with one mutation} has been introduced and studied in~\cite{sentis2016quantum} for the minimum error case and in~\cite{sentis2017exact} in the unambiguous case. 
In both works, they present the optimal success probabilities and analyze restricted measurements given by online algorithms. 
Achieving the bound in~\cite{sentis2016quantum} requires a measurement performed jointly on all states in the sequence. 
In this paper, we make the same assumption, that we have the availability of quantum memory capable of storing each of these particles sent by Alice, thereby enabling such joint measurements.


\subsubsection{Application 1: A single changepoint with varying reward functions}  


\paragraph{One changepoint, with horseshoe reward.}   

In this example, we consider the horseshoe reward with varying closeness parameters $\mu$. 
Recall the reward function is given as 
\begin{align}
    \label{eq:blasting}
    R_{ij} =
    \begin{cases}
        1 &\text{if} \quad |i-j| \le \mu \\
        0 &\text{otherwise}
    \end{cases}.    
\end{align}

Figure~\ref{fig:beta'_beta''_pi_4_1CP_mu} plots the optimal reward value for a single changepoint example. 


\begin{figure}[h]
    \centering
    \includegraphics[scale=0.75]{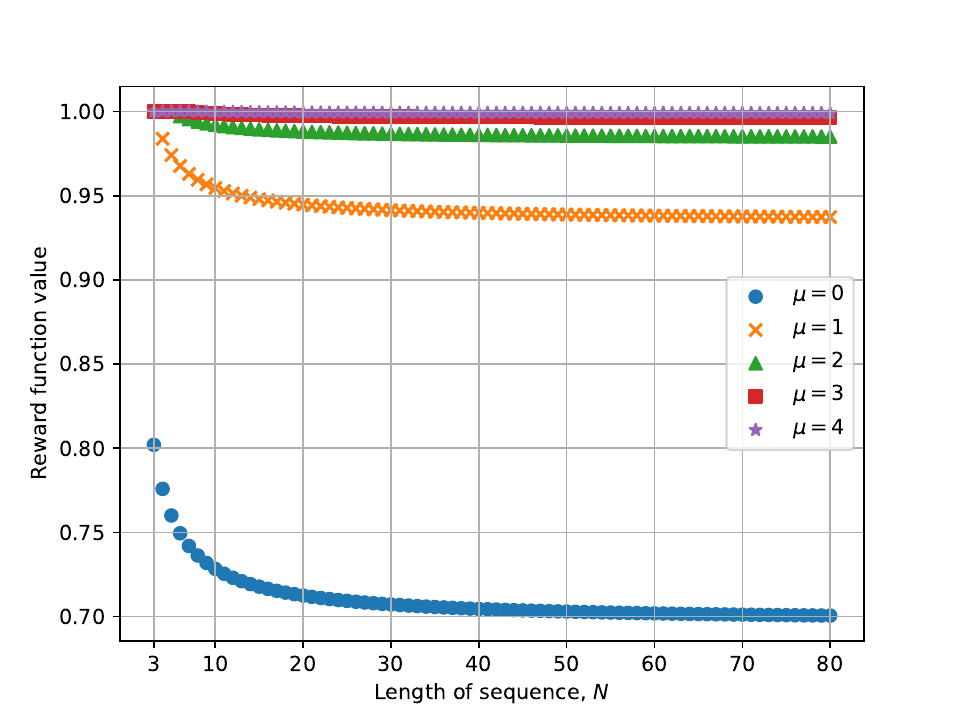}
    \caption{
    We look at a sequence of $N$ states starting at $\ket{0}$ and possibly switching to $\ket{+}$ at some point.
    We look at the optimal reward function value (which corresponds to the optimal probability of identifying the location of the changepoint) as a function of $N$ for varying values of closeness parameter $\mu$. 
    We see that the reward stabilizes for large sequence lengths and having large values of $\mu$ greatly improves the reward, as expected.
    }
    \label{fig:beta'_beta''_pi_4_1CP_mu}
\end{figure} 

We note that in Figure~\ref{fig:beta'_beta''_pi_4_1CP_mu}, we were easily able to compute the reduced SDPs for this problem up to $80$ qubits. 
Using the original SDP, this would be completely intractable as the variables would each be of size $2^{80} \times 2^{80}$.


\paragraph{One changepoint, with the closer-the-better reward.}

In this example, we consider the closer-the-better reward with varying  parameters $\gamma$. 
Recall the reward function is given as 
\begin{align*}
    R_{ij} = \gamma^{|i-j|}.
\end{align*}

An example reward function is illustrated in Figure~\ref{fig:overall_1CP} (a) while Figure~\ref{fig:overall_1CP} (b) plots the optimal reward value for three particular examples.

\begin{figure}[h]
    \centering
    \begin{subfigure}[b]{0.495\textwidth}
        \centering
        \includegraphics[width=\textwidth]{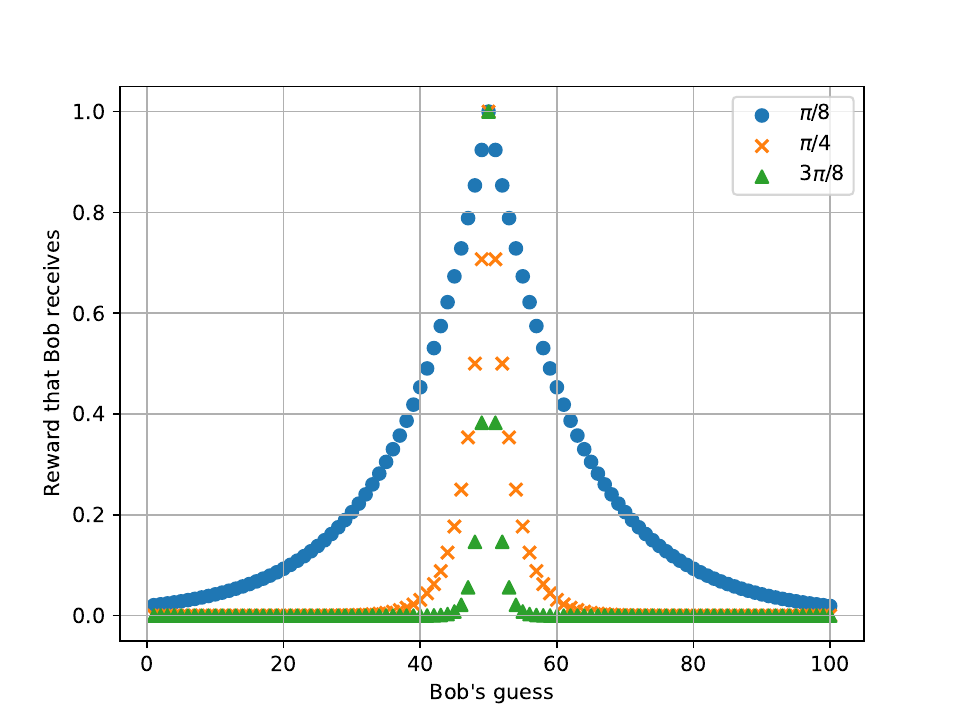}
        \caption{Reward that Bob receives given the sequence $\ket{\psi}^{\otimes 50} \otimes \ket{\phi}^{\otimes 50}$ for the closer-the-better reward scheme.}
    \end{subfigure}
    \hfill
    \begin{subfigure}[b]{0.495\textwidth}
        \centering
        \includegraphics[width=\textwidth]{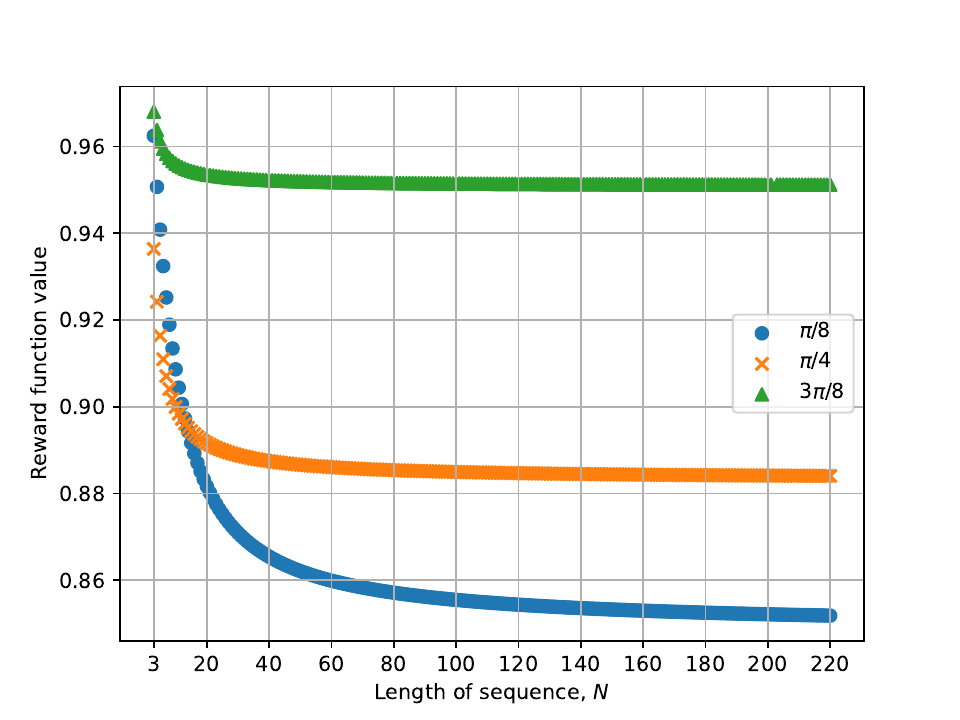}
        \caption{Optimal reward function value as a function of $N$ for varying values of $\theta$.}
    \end{subfigure}
    \caption{
    We consider here the closer-the-better reward scheme with $\gamma = |\braket{\psi}{\phi}| = \cos(\theta)$ for {$\theta \in \{ \pi/8, \pi/4, 3\pi/8 \}$} when $\ket{\psi} = \ket{0}$ and $\ket{\phi} = \cos(\theta) \ket{0} + \sin(\theta) \ket{1}$.
    The figure on the left describes the reward that Bob receives when given the sequence $\ket{\psi}^{\otimes 50} \otimes  \ket{\phi}^{\otimes 50}$ for this reward scheme, while the figure on the right depicts the optimal reward function as a function of $N$, the number of states sent.
    Note here that the closer the states are to each other, the harder it is to locate the changepoint.
    However, we scaled the rewards such that we have larger rewards in the more difficult settings.
    }
    \label{fig:overall_1CP}
\end{figure}

In this setting, we are able to easily compute the reward function for states up to $220$ qubits.

\subsubsection{Heuristic: Speeding up computations for large changepoint sequences}  

To speed-up the calculation of the optimal reward values for large sequences, we introduce a heuristic. 
Roughly speaking, we add the constraint in the dual of the reduced problem that $X$ is Hermitian Toeplitz (see Section~\ref{sec:bg} for a discussion of Toeplitz matrices). 
The dual is given as $\beta'$ below and $\beta''$ is our heuristic.
\begin{align}
    \beta' & = \text{minimize:} 
        \left\{ 
            \hsip{X}{G} : 
            X \succcurlyeq \sum_{j=1}^N R_{ij} q_j \ketbra{j}, \ i \in [L]
        \right\} \label{reduceddual}
\end{align}
\begin{equation}
    \beta'' = \text{minimize:} 
                \left\{ 
                    \hsip{X}{G} : 
                    X \succcurlyeq \sum_{j=1}^N R_{ij} q_j \ketbra{j}, \ i \in [L], 
                    X \text{ Hermitian Toeplitz}
                \right\}. \label{reduceddual2}
\end{equation}

By strong duality (see Appendix~\ref{sec:bg}), the primal (Eq.~\eqref{reducedprimal}) and the dual (Eq.~\eqref{reduceddual}) have the same optimal value. 

Note that we justify this heuristic in the specific application of the quantum changepoint problem with one changepoint; there is no reason to believe it works well in general discrimination problems. 

Why this heuristic? Well, for one it reduces the number of free parameters in an $N \times N$ matrix to just $N$. 
But, this means nothing if it gives us bad approximations. 
Indeed, it turns out that for large $N$, this heuristic turns out to both be a good approximation to the actual value as well as seeing an advantage in computation time. 
In Figure~\ref{fig:beta'_beta''_pi_4_1CP_intro} we consider the case where Alice promises to send Bob $N$ copies of the state $\ket{0}$ but a mutation in her device led to the generation of the state $\ket{+}$ instead.
The reduced SDP (Eq.~\eqref{reduceddual}) and the heuristic (Eq.~\eqref{reduceddual2}) are considered.
We compare the absolute difference of their values as well as their runtimes respectively.
For a sequence of length $220$, the error between the reduced SDP (Eq.~\eqref{reduceddual}) and the heuristic (Eq.~\eqref{reduceddual2}) is of the order of $10^{-3}$ and computing the heuristic (Eq.~\eqref{reduceddual2}) is observed to be about \emph{seven} times faster than the reduced SDP (Eq.~\eqref{reduceddual}).\footnote{In Figure~\ref{fig:beta'_beta''_pi_4_1CP_intro}, the subscript $1$ in $\beta_1'$ and $\beta_1''$ denote that we are considering one changepoint. We consider $2$ and $3$ changepoints in the appendix.}. 

Appendix~\ref{sec:mQCP} describes the general case and provides a comparison with the results of~\cite{sentis2016quantum, sentis2017exact}. 

\begin{figure}[h]
    \centering
    \begin{subfigure}[b]{0.495\textwidth}
        \centering
        \includegraphics[width=\textwidth]{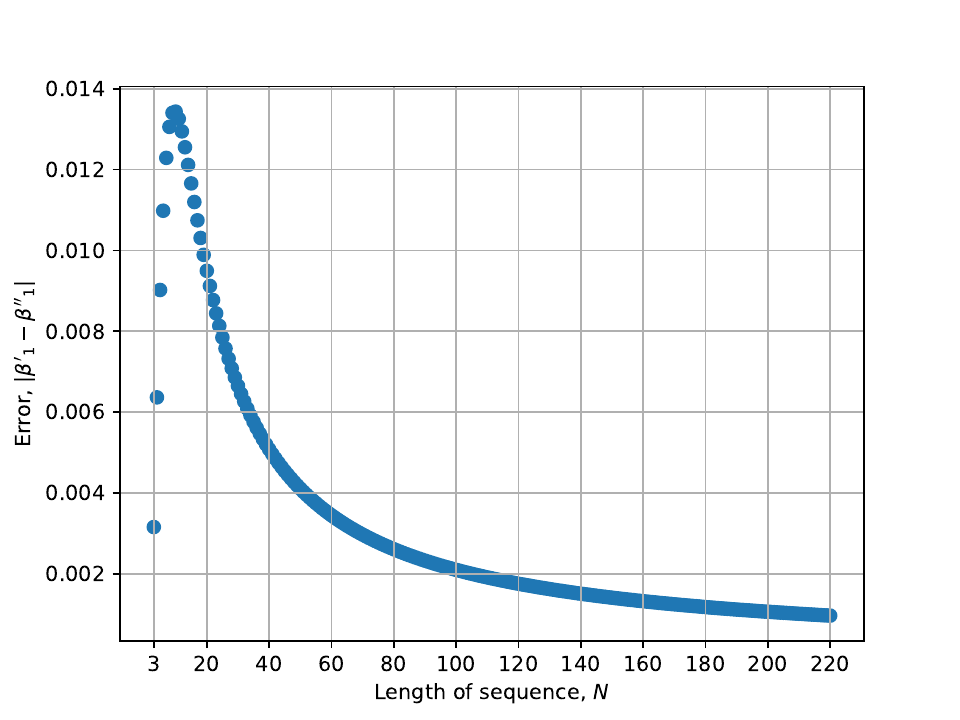}
        \caption{Absolute value of the difference between Eq.~\eqref{reduceddual} and Eq.~\eqref{reduceddual2}.}
    \end{subfigure}
    \hfill
    \begin{subfigure}[b]{0.495\textwidth}
        \centering
        \includegraphics[width=\textwidth]{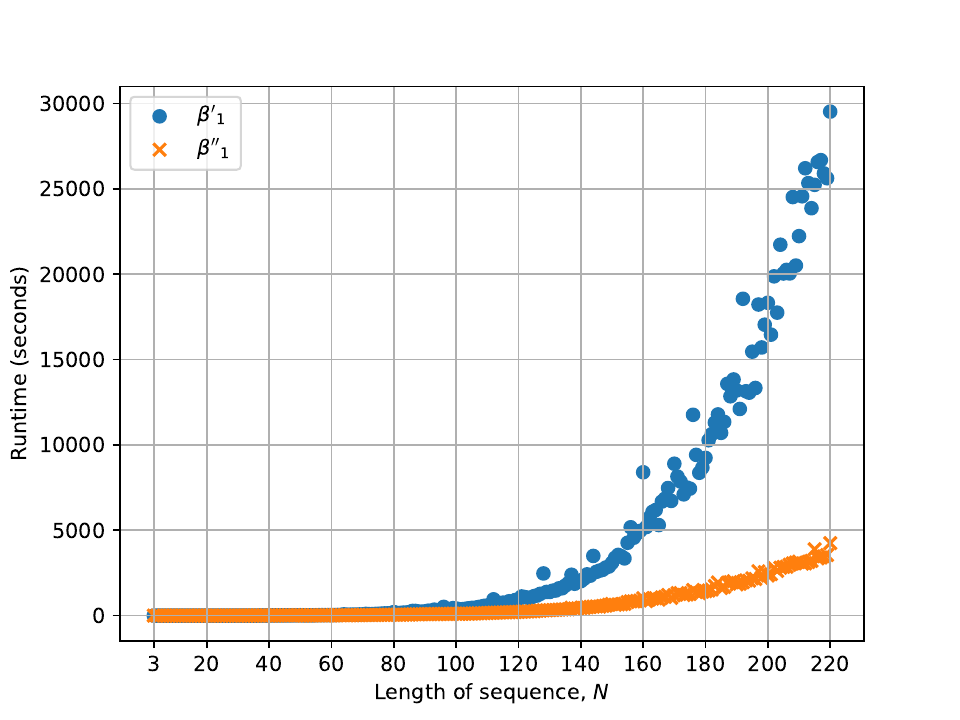}
        \caption{Runtimes of Eq.~\eqref{reduceddual} and Eq.~\eqref{reduceddual2}.}
    \end{subfigure}
    \caption{
    Here we consider the sequences where the state $\ket{0}$ mutated to the state $\ket{+}$. 
    Using the closer-the-better reward scheme with $\gamma = 1/2$, we see the error between the heuristic (Eq.~\eqref{reduceddual2}) and the reduced SDP (Eq.~\eqref{reduceddual}) (on the left), and the runtimes of both Eq.~\eqref{reduceddual} and Eq.~\eqref{reduceddual2} (on the right).
    Observe that the difference between Eq.~\eqref{reduceddual} and Eq.~\eqref{reduceddual2} decreases rapidly approaching $10^{-3}$ at $N = 220$ (on the left).
    Although there seems to be no advantage, in terms of time, in computing Eq.~\eqref{reduceddual2} for sequences up to length $80$, we observe that computing Eq.~\eqref{reduceddual2} is roughly \emph{seven} times faster than calculating Eq.~\eqref{reduceddual} for larger values of $N$.
    }
    \label{fig:beta'_beta''_pi_4_1CP_intro}
\end{figure}


\subsubsection{Application 2: Multiple changepoint problem}  

We now consider the scenario where Alice promises $N$ copies of the state $\ket{0}$, but there are a total of three changepoints; starting from $\ket{0}$ and eventually possibly changing to $\ket{+}$, to $\ket{1}$, and to $-\ket{-}$, in that order.
Figure~\ref{fig:beta'_beta''_pi_4_3CP} (a) and Figure~\ref{fig:beta'_beta''_pi_4_3CP} (b) compare the error and their runtimes of the reduced dual (Eq.~\eqref{reduceddual}) and the heuristic (Eq.~\eqref{eq:T_3CP}) respectively.
These examples and more are described in further detail in Appendix~\ref{sec:expts}.

\subsection{Quantum error type classification}

Suppose now that Alice sends Bob an $n$-qubit state $\ket{\phi}$ via a single-qubit Pauli error channel.
We assume that any of the $n$ qubits can be affected by one of these errors with equal probability.
Bob's task is to simply find out which of these errors occurred, i.e., $I, X, Z$, or $ZX$ (which equals $Y$ up to global phase).

\begin{figure}
    \centering
    \includegraphics[width=\textwidth]{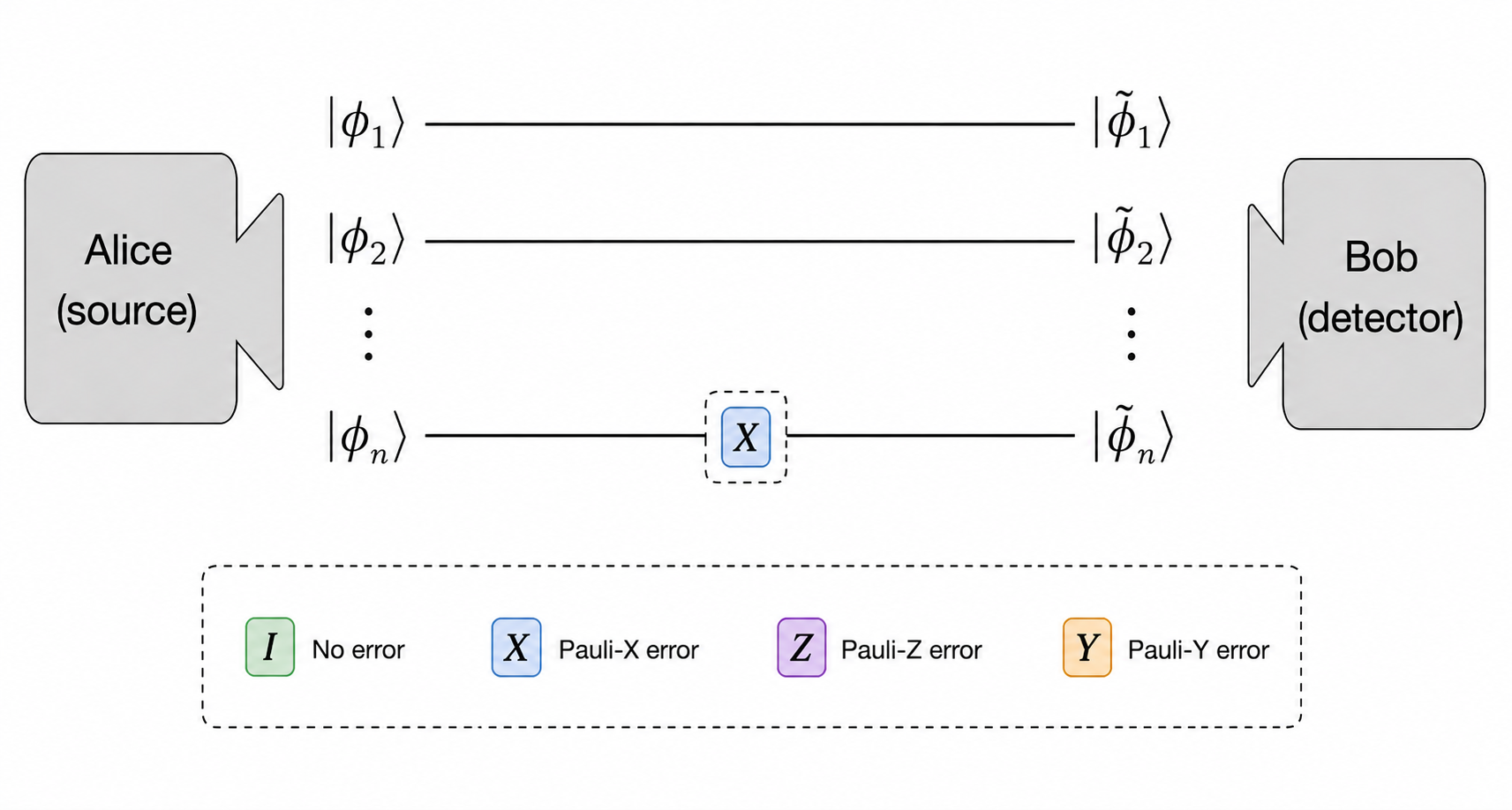}
    \caption{A quantum source is sending an $n$-qubit state $\ket{\phi}$ via a single-qubit Pauli error channel. The task is to classify the type of error ($I, X, Z, Y$) that affected the state (while not asking for \emph{which qubit} it acted on).}
    \label{fig:QETC}
\end{figure}

For the $n$-qubit state, all of the $N = 3n+1$ possibilities are 
\begin{equation}
    \mathcal{S} \coloneqq \{ \ket{\phi}, X_1 \ket{\phi}, \ldots, X_n \ket{\phi}, Z_1 \ket{\phi}, \ldots, Z_n \ket{\phi}, Z_1 X_1 \ket{\phi}, \ldots, Z_n X_n \ket{\phi} \}.    
\end{equation}
Bob can formulate this as a problem of the form described in Eq.~\eqref{eq:mixed_states} with the reward given in Eq.~\eqref{eq:class}, as follows.
Define the set $\mathcal{I} = \{ I, X, Z, ZX \}$ and let $C_i$ for $i \in \mathcal{I}$ denote that Pauli-$i$ error occurred.
Let $R_{ij} = 1$ if $\ket{\psi_j} \in \mathcal{S}$ is correctly classified into $C_i$, and $0$ otherwise.
The a priori probability of the states are $q_1 = 1/4$ and $1/4n$ for each of the other states in $\mathcal{S}$.
Thus, the optimal classification probability can be formulated as the following SDP
\begin{equation}
    \max \; \left\{ \sum_{i \in \mathcal{I}} \sum_{j=1}^N R_{ij} q_j \hsip{M_i}{\ketbra{\psi_j}}\ : \sum_{i \in \mathcal{I}} M_i = I,\ M_i \succcurlyeq 0 \right\}.
\end{equation}

For our numerical experiments, we consider the state
\begin{equation}
    \label{eq:qetc}
    \ket{\psi_n(\theta)} = \cos(\theta) \ket{\mathrm{GHZ}_n} + \sin(\theta) \ket{\mathrm{W}_n},
\end{equation}
where
\begin{equation*}
    \ket{\mathrm{GHZ}_n} = \frac{\ket{0}^{\otimes n} + \ket{1}^{\otimes n}}{\sqrt{2}} \quad \text{and} \quad \ket{\mathrm{W}_n} = \frac{1}{\sqrt{n}}(\ket{100 \cdots 0} + \ket{010 \cdots 0} + \ket{001 \cdots 0} + \cdots + \ket{000 \cdots 1}) 
\end{equation*} 
for $n \geq 2$.

For $n=2$, the Bell states are $\ket{\phi^\pm} = (\ket{00} \pm \ket{11})/\sqrt{2}$ and $\ket{\psi^\pm} = (\ket{01} \pm \ket{10})/\sqrt{2}$. 
We get $\ket{\mathrm{GHZ}_2} = \ket{\phi^+}$ and $\ket{\mathrm{W}_2} = \ket{\psi^+}$.
Let us define the state $\ket{\chi(\theta)^\pm} = \cos(\theta) \ket{\phi^\pm} + \sin(\theta) \ket{\psi^\pm}$, $\Phi^\pm = \ketbra{\phi^\pm}$, $\Psi^\pm = \ketbra{\psi^\pm}$, and let $\bar{\theta} = \pi/2 - \theta$.

The state $\ket{\psi_2(\theta)} = \ket{\chi(\theta)^+}, X_1\ket{\psi_2(\theta)} = X_2\ket{\psi_2(\theta)} = \ket{\chi(\bar{\theta})^+}$.
We have $Z_1\ket{\psi_2(\theta)} = \ket{\chi(\theta)^-}$, $Z_2\ket{\psi_2(\theta)} = \ket{\chi(-\theta)^-}$, $Z_1 X_1 \ket{\psi_2(\theta)} = \ket{\chi(\bar{\theta})^-}$ and $Z_2 X_2 \ket{\psi_2(\theta)} = \ket{\chi(-\bar{\theta})^-}$.

Define
\begin{equation*}
    \begin{aligned}
        \rho_I &\coloneqq \ketbra{\chi(\theta)^+} = \cos^2(\theta)\ \Phi^+ + \sin^2(\theta)\ \Psi^+ + \frac{\sin(\theta)\cos(\theta)}{2}\ (\id \otimes X + X \otimes \id), \\
        \rho_X &\coloneqq \ketbra{\chi(\bar\theta)^+} = \sin^2(\theta)\ \Phi^+ + \cos^2(\theta)\ \Psi^+ + \frac{\sin(\theta)\cos(\theta)}{2} (\id \otimes X + X \otimes \id), \\
        \rho_Z &\coloneqq \frac{1}{2} \ketbra{\chi(\theta)^-} + \frac{1}{2} \ketbra{\chi(-\theta)^-} = \cos^2(\theta)\ \Phi^- + \sin^2(\theta)\ \Psi^-, \\
        \rho_{ZX} &\coloneqq \frac{1}{2} \ketbra{\chi(\bar{\theta})^-} + \frac{1}{2} \ketbra{\chi(-\bar{\theta})^-} = \sin^2(\theta)\ \Phi^- + \cos^2(\theta)\ \Psi^-.
    \end{aligned}
\end{equation*}

The problem can be expressed as
\begin{equation}
    \label{eq:n_2_form}
    \begin{aligned}
        P_{\mathrm{succ}} = \max \left\{ \frac{1}{4} \hsip{M_I}{\rho_I} + \frac{1}{4} \hsip{M_X}{\rho_X} + \frac{1}{4} \hsip{M_Z}{\rho_Z} + \frac{1}{4} \hsip{M_{ZX}}{\rho_{ZX}}\ \right\}.
    \end{aligned}
\end{equation}

Note that $\rho_I, \rho_X \in \mathrm{span}(\Phi^+, \Psi^+)$ and $\rho_Z,\rho_{ZX} \in \mathrm{span}(\Phi^-, \Psi^-)$.
Thus, we can write Eq.~\eqref{eq:n_2_form} as $P_{\mathrm{succ}} = P_{\mathrm{succ}}^+ + P_{\mathrm{succ}}^-$, where 
\begin{equation*}
    \label{eq:n_2_div}
    \begin{aligned}
        P_{\mathrm{succ}}^+ &= \max \left\{ \frac{1}{4} \hsip{M_I}{\rho_I} + \frac{1}{4} \hsip{M_X}{\rho_X}\ :\ M_I + M_X = \Pi_+,\ M_I, M_X \succcurlyeq 0 \right\} \\
        P_{\mathrm{succ}}^- &= \max \left\{ \frac{1}{4} \hsip{M_Z}{\rho_Z} + \frac{1}{4} \hsip{M_{ZX}}{\rho_{ZX}}\ :\ M_Z + M_{ZX} = \Pi_-,\ M_Z, M_{ZX} \succcurlyeq 0 \right\}
    \end{aligned}
\end{equation*}
with $\Pi_\pm = \Phi^\pm + \Psi^\pm$.
It is straightforward to see that $\Pi_+ + \Pi_- = \id$.

Consider $P_{\mathrm{succ}}^+$.
Let us set $M_X = \Pi_+ - M_I$.
Then
\begin{equation*}
    P_{\mathrm{succ}}^+ = \frac{1}{4} \hsip{M_I}{\rho_I - \rho_X} + \frac{1}{4} \hsip{\Pi_+}{\rho_X} = \frac{1}{4} + \frac{\cos(2 \theta)}{4} \hsip{M_I}{\Phi^+ - \Psi^+}.
\end{equation*}
Similarly, by setting $M_{ZX} = \Pi_- - M_Z$, we get
\begin{equation*}
    P_{\mathrm{succ}}^- = \frac{1}{4} + \frac{\cos(2 \theta)}{4} \hsip{M_Z}{\Phi^- - \Psi^-}.
\end{equation*}
Thus, Eq.~\eqref{eq:n_2_form} becomes
\begin{equation*}
    P_{\mathrm{succ}} = \frac{1}{2} + \frac{\cos(2 \theta)}{4} \left( \hsip{M_I}{\Phi^+ - \Psi^+} + \hsip{M_Z}{\Phi^- - \Psi^-} \right)
\end{equation*}

For $\theta < \pi/4$, we have $\cos(2 \theta) > 0$ and so by choosing $M_I = \Phi^+$ and $M_Z = \Phi^-$ we get $P_\mathrm{succ} = (1 + \cos(2 \theta))/2 = \cos^2(\theta)$.
Similarly, for $\theta > \pi/4$, we have $\cos(2 \theta) < 0$ and so by choosing $M_I = \Psi^+$ and $M_Z = \Psi^-$ we get $P_\mathrm{succ} = (1 - \cos(2 \theta))/2 = \sin^2(\theta)$.
At $\theta = \pi/4$, we have $\cos(2 \theta) = 0$ which gives us $P_\mathrm{succ} = 1/2$.
Combining these, we get $P_\mathrm{succ} = \max\{ \cos^2(\theta), \sin^2(\theta) \}$.

While finding an analytical expression for the SDP value for $n=2$ is not too hard, this is much harder for larger values of $n$ and thus we turn to computing these numerically. 
Figure~\ref{fig:qetc} illustrates the optimal error type classification probability as a function of the angle $\theta$ of the state in Eq.~\eqref{eq:qetc}. 

To further examine the classification probabilities at $\theta = 49.8^\circ$, where there is seemingly an inflection point, we calculated the success probability for larger values of $n$ at this value of $\theta$. 
We observe that the success probability keeps shifting very slightly as the value of $n$ increases, with a probability of $0.906102$ at $n = 300$ which took $249891$ seconds (roughly $2.89$ days) of computation time (refer to the computational platform section for hardware details). 
Using the original SDP, this would correspond to dealing with variables of size $2^{300} \times 2^{300}$ which is computationally infeasible. 

\begin{remark} 
Note that this classical simulation was feasible since the state overlaps had an exploitable structure. 
However, for generic states and error unitaries, one would likely need to take advantage of a quantum computer to compute the Gram matrix. 
\end{remark} 

\begin{figure}
    \centering
    \includegraphics[scale=0.5]{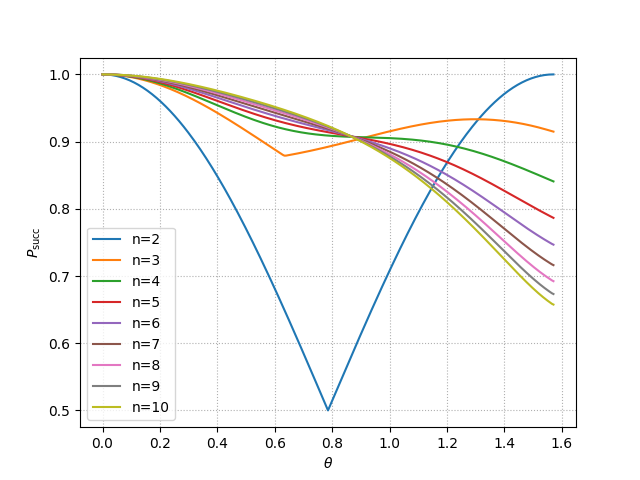}
    \caption{
    The optimal error type classification probability of ${\ket{\psi_n(\theta)} = \cos(\theta) \ket{\mathrm{GHZ}_n} + \sin(\theta) \ket{\mathrm{W}_n}}$ as a function of $\theta$ as it varies from $0$ to $\pi/2$.
    For $n=2$, the success probability is $\max\{\cos^2(\theta), \sin^2(\theta)\}$. 
    For $n \geq 3$, we observe that there seems to be an inflection point around $\theta = 49.8^\circ$ corresponding to a success probability of $0.907$.
    }
    \label{fig:qetc}
\end{figure}

\section{Reducing the size of the reward SDPs}\label{sec:SDPs}

The optimal reward value of the general Bayes approach, given in Eq.~\eqref{eq:gen_QSD2} can be formally described as the primal-dual pair, below 
\begin{equation}
\begin{minipage}{.45\textwidth}
    \begin{equation*}
        \begin{aligned}
            \alpha = \max \; & \sum_{i=1}^L \sum_{j=1}^N R_{ij} q_j \hsip{M_i}{\rho_j} \\
            \st \; & \sum_{i=1}^L M_i = \id \\
            & M_i \succeq 0, \, i \in [L],
        \end{aligned}
    \end{equation*}    
\end{minipage} 
\begin{minipage}{.45\textwidth} 
    \begin{equation*}
        \begin{aligned}
            \beta = \min \; & \Tr(Y) \\
            \st \; & Y \succcurlyeq \sum_{j=1}^N R_{ij} q_j\, \rho_j, \, i \in [L]. \\
            & \phantom{N} \\ 
            & \phantom{N} 
        \end{aligned}
    \end{equation*}
\end{minipage}
\label{eq:alpha_beta}
\end{equation} 
where we use $[m] := \{ 1, \ldots, m \}$ for brevity. 
Note that $\alpha = \beta$ from strong duality when $R_{ij}$ is finite.

We consider the case when the states are pure, i.e., $\{\ket{\psi_1}, \dots, \ket{\psi_N}\}$, with corresponding a priori probabilities $\{q_1, \dots, q_N\}$. 
This gives us the SDPs 
\begin{equation}
\begin{minipage}{.45\textwidth}
    \begin{equation*}
        \begin{aligned}
            \alpha = \max \; & \sum_{i=1}^L \sum_{j=1}^N R_{ij} q_j \hsip{M_i}{\kb{\psi_j}} \\
            \st \; & \sum_{i=1}^L M_i = \id \\
            & M_i \succeq 0, \, i \in [L], 
        \end{aligned}
    \end{equation*}    
\end{minipage} 
\begin{minipage}{.45\textwidth} 
    \begin{equation*}
        \label{prog:P_D}
        \begin{aligned}
            \beta = \min \; & \Tr(Y) \\
            \st \; & Y \succcurlyeq \sum_{j=1}^N R_{ij} q_j \kb{\psi_j}, \, i \in [L]. \\
            & \phantom{N}
        \end{aligned}
    \end{equation*}
\end{minipage}
\end{equation} 

We now suppose a particular structure on the variables in the primal above. 
We show that there is no loss of generality in considering this structure. 
To this end, define the following matrix 
\begin{equation}
    \Psi \coloneqq \sum\limits_{k=1}^N \ketbra{\psi_k}{k} 
\end{equation} 
and note that $\Psi^*\Psi = G$, the Gram matrix of the set of states $\{ \ket{\psi_1}, \ldots, \ket{\psi_N} \}$. 
For a primal feasible solution $\{M_1, \dots, M_L\}$, define $W_i := \Psi^* M_i \Psi$ for $i \in \{1, \dots, L\}$, since $M_i \succeq 0$ for all $i \in [L]$ and $\sum_{i=1}^L M_i = \id$, we have that 
\begin{equation} 
    W_i \succeq 0, \text{ for all } i \in [L], \text{ and } \sum_{i=1}^L W_i = G. 
\end{equation}
Also, we have $\ket{\psi_j} = \Psi \ket{j}$, so we have
\begin{equation}\label{eq:M_W}
    \hsip{M_i}{\kb{\psi_j}} 
    = \bra{\psi_j} {M_i} \ket{\psi_j} 
    = \bra{j} \Psi^* {M_i} \Psi \ket{j} 
    = \bra{j} W_i \ket{j}. 
\end{equation}
Thus, under this assumption, we have the following primal-dual pair of SDPs
\begin{equation}
\begin{minipage}{.4\textwidth}
    \begin{equation*}
        \begin{aligned}
            \alpha' = \max \; & \sum_{i=1}^L \sum_{j=1}^N R_{ij} q_j \bra{j} W_i \ket{j} \\
            \st \; & \sum_{i=1}^L W_i = G \\
            & \hspace{20pt} W_i \succeq 0,\ i \in [L]. \\ 
        \end{aligned}
    \end{equation*}
\end{minipage}
~
\begin{minipage}{.4\textwidth}
    \begin{equation*}
        \label{prog:P'_D'}
        \begin{aligned}
            \beta' = \min \quad & \hsip{X}{G} \\
            \st \quad & X \succcurlyeq \sum_{j=1}^N R_{ij} q_j \ketbra{j},\ i \in [L]. \\
            & \quad \\ 
            & \quad \\ 
        \end{aligned}
    \end{equation*}    
\end{minipage}
\end{equation} 

Since we can take a feasible $(M_1, \ldots, M_L)$ and construct feasible $(W_1, \ldots, W_L)$ with the same objective function value, we have that $\alpha \leq \alpha'$.
By strong duality, one can check that $\alpha' = \beta'$.
All that remains to be proved is that $\alpha \geq \alpha'$.

Suppose $\{W_1, \dots, W_L\}$ is a feasible solution of the primal SDP in Eq.~\eqref{eq:M_W}.
If $\Psi = \sum_{i=1}^r \sigma_i \ketbra{\mu_i}{\nu_i}$ is the singular value decomposition of $\Psi$ where $r$ denotes its rank, then $\Psi^+ = \sum_{i=1}^r \sigma_i^{-1} \ketbra{\nu_i}{\mu_i}$ is the Moore-Penrose generalized inverse of $\Psi$.
Define $M_i = (\Psi^+)^*\ W_i\ (\Psi^+) + \frac{1}{L} \left( \id - \Psi \Psi^+ \right)$ for each $i \in [L]$.
Since $W_i \succeq 0$ and $\Psi \Psi^+ = \sum_{i=1}^r \ketbra{\mu_i}$ is a projection, we have that $M_i \succeq 0$ for each $i \in [L]$.
If we sum over all the elements $\{M_1, \dots, M_L\}$, we get
\begin{align}
    \sum_{i=1}^L M_i 
    &= \sum_{i=1}^L \left( (\Psi^+)^*\ W_i\ (\Psi^+) + \frac{1}{L} \left( \id - \Psi \Psi^+ \right) \right) \nonumber \\
    &= (\Psi^+)^*\ \left( \sum_{i=1}^L W_i \right) \Psi^+ + \left( \sum_{i=1}^L \frac{1}{L} \right) \left( \id - \Psi \Psi^+ \right) = \id.
\end{align}
The last equality is obtained from noting that $\sum_{i=1}^L W_i = G = \Psi^* \Psi$.
For each value of $i$, we have 
\begin{align}
    \bra{\psi_j} M_i \ket{\psi_j}
    = \bra{j} \Psi^* M_i \Psi \ket{j}
    = \bra{j} \Psi^* \left( (\Psi^+)^*\ W_i\ (\Psi^+) + \frac{1}{L} \left( \id - \Psi \Psi^+ \right) \right) \Psi \ket{j}
    = \bra{j} W_i \ket{j}.
\end{align} 
The last equality holds because $\sum_{i=1}^L W_i = G = \Psi^* \Psi$, implying that
\begin{align}
    \text{Range}(W_i) \subseteq \text{Range}(G) = \text{Range}(\Psi^*)    
\end{align}
and $\Psi^+ \Psi$ is the projection onto the range of $\Psi^*$ (and thus leaves $W_i$ unchanged).
Therefore, we get $\alpha \ge \alpha'$.

\medskip 
Note that this proof is similar in structure to a reduction in~\cite{johnston2023tight}.
In that work, the focus was on quantum state exclusion.
The proof above is the generalization to the general Bayes approach.
   
\section{Conclusion and future work} 
In this work we studied Holevo's Bayes approach for quantum state discrimination and provided a new SDP reduction which is fully characterized by the Gram matrix of the states to be guessed.  
We also discussed how the Gram matrix can be computed given the source code (preparation circuits) making the dimension-reduction pre-computation phase efficiently implementable on a quantum computer.
We applied this procedure to the quantum changepoint identification problem under various reward settings and also a quantum error classification task where in each we were able to simulate our procedure for hundreds of qubits. 

Potential future directions include running this procedure on real quantum hardware where the Gram matrix of the outputs of quantum circuits cannot be computed so easily. 
Moreover, since this framework is so general, many other quantum state discrimination problems can also now be considered, especially those with nicely-structured Gram matrices.  

\section*{Computational platform}

For the quantum changepoint problems, the SDPs were solved using CVX~\cite{cvx} on a $32$ GB $10$th generation Intel Core i$9$-$10885$H CPU ($16$ MB cache, $2.40$ GHz, $8$ cores).
For the quantum error type identification problem, the SDPs were solved using CVXPY~\cite{diamond2016cvxpy, agrawal2018rewriting} on a $32$ GB AMD Ryzen $9$ $3950$X $(32)$ @ $3.50$ GHz CPU with NVIDIA GeForce RTX $2080$ Ti Rev. A GPU.
 
\section*{Acknowledgement}

This work was performed when A.M. was an intern at Fujitsu Research of America. 
J.S. acknowledges support from a Fujitsu Research of America research grant and through the NSF Award 2542721.  

\bibliographystyle{alpha}
\bibliography{apssamp} 

@PREAMBLE{
 "\providecommand{\noopsort}[1]{}" 
 # "\providecommand{\singleletter}[1]{#1}%" 
}

@article{sentis2016quantum,
  title={Quantum change point},
  author={Sent{\'\i}s, Gael and Bagan, Emilio and Calsamiglia, John and Chiribella, Giulio and Munoz-Tapia, Ramon},
  journal={Physical Review Letters},
  volume={117},
  number={15},
  pages={150502},
  year={2016},
  publisher={APS}
}

@article{sentis2017exact,
  title={Exact identification of a quantum change point},
  author={Sent{\'\i}s, Gael and Calsamiglia, John and Munoz-Tapia, Ramon},
  journal={Physical Review Letters},
  volume={119},
  number={14},
  pages={140506},
  year={2017},
  publisher={APS}
}

@article{sentis2018online,
  title={Online strategies for exactly identifying a quantum change point},
  author={Sent{\'\i}s, Gael and Mart{\'\i}nez-Vargas, Esteban and Munoz-Tapia, Ramon},
  journal={Physical Review A},
  volume={98},
  number={5},
  pages={052305},
  year={2018},
  publisher={APS}
}

@article{skotiniotis2024identification,
  title={Identification of malfunctioning quantum devices},
  author={Skotiniotis, Michalis and Llorens, Santiago and Hotz, Ronja and Calsamiglia, John and Mu{\~n}oz-Tapia, Ramon},
  journal={Physical review research},
  volume={6},
  number={3},
  pages={033329},
  year={2024},
  publisher={APS}
}

@article{martinez2019certified,
  title={Certified answers for ordered quantum discrimination problems},
  author={Mart{\'\i}nez-Vargas, Esteban and Munoz-Tapia, Ramon},
  journal={arXiv preprint arXiv:1908.04093},
  year={2019}
}

@article{russo2023inner,
  title={Inner products of pure states and their antidistinguishability},
  author={Russo, Vincent and Sikora, Jamie},
  journal={Physical Review A},
  volume={107},
  number={3},
  pages={L030202},
  year={2023},
  publisher={APS}
}

@article{helstrom1969quantum,
  title={Quantum detection and estimation theory},
  author={Helstrom, Carl W},
  journal={Journal of Statistical Physics},
  volume={1},
  pages={231--252},
  year={1969},
  publisher={Springer}
}

@article{barnett2001minimum,
  title={Minimum-error discrimination between multiply symmetric states},
  author={Barnett, Stephen M},
  journal={Physical Review A},
  volume={64},
  number={3},
  pages={030303},
  year={2001},
  publisher={APS}
}

@article{ducuara2019operational,
  title={Operational interpretation of weight-based resource quantifiers in convex quantum resource theories of states},
  author={Ducuara, Andr{\'e}s F and Skrzypczyk, Paul},
  journal={arXiv preprint arXiv:1909.10486},
  year={2019}
}

@article{chiribella2004covariant,
  title={Covariant quantum measurements that maximize the likelihood},
  author={Chiribella, Giulio and D’Ariano, Giacomo Mauro and Perinotti, Paolo and Sacchi, Massimiliano F},
  journal={Physical Review A},
  volume={70},
  number={6},
  pages={062105},
  year={2004},
  publisher={APS}
}

@article{chiribella2006maximum,
  title={Maximum likelihood estimation for a group of physical transformations},
  author={Chiribella, Giulio and D’Ariano, Giacomo Mauro and Perinotti, Paolo and Sacchi, Massimiliano F},
  journal={International Journal of Quantum Information},
  volume={4},
  number={03},
  pages={453--472},
  year={2006},
  publisher={World Scientific}
}

@article{chefles1998unambiguous,
  title={Unambiguous discrimination between linearly independent quantum states},
  author={Chefles, Anthony},
  journal={Physics Letters A},
  volume={239},
  number={6},
  pages={339--347},
  year={1998},
  publisher={Elsevier}
}

@article{llorens2024quantum,
  title={Quantum multi-anomaly detection},
  author={Llorens, Santiago and Sent{\'\i}s, Gael and Mu{\~n}oz-Tapia, Ramon},
  journal={Quantum},
  volume={8},
  pages={1452},
  year={2024},
  publisher={Verein zur F{\"o}rderung des Open Access Publizierens in den Quantenwissenschaften}
}

@article{helstrom2003bayes,
  title={Bayes-cost reduction algorithm in quantum hypothesis testing (corresp.)},
  author={Helstrom, C},
  journal={IEEE Transactions on Information Theory},
  volume={28},
  number={2},
  pages={359--366},
  year={2003},
  publisher={IEEE}
}

@article{holevo1973statistical,
  title={Statistical decision theory for quantum systems},
  author={Holevo, Alexander S},
  journal={Journal of multivariate analysis},
  volume={3},
  number={4},
  pages={337--394},
  year={1973},
  publisher={Elsevier}
}

@article{holevo1978asymptotically,
  title={On asymptotically optimal hypotheses testing in quantum statistics},
  author={Holevo, Alexander Semenovich},
  journal={Teoriya Veroyatnostei i ee Primeneniya},
  volume={23},
  number={2},
  pages={429--432},
  year={1978},
  publisher={Russian Academy of Sciences, Steklov Mathematical Institute of Russian~…}
}

@article{yuen2003optimum,
  title={Optimum testing of multiple hypotheses in quantum detection theory},
  author={Yuen, Horace and Kennedy, Robert and Lax, Melvin},
  journal={IEEE transactions on information theory},
  volume={21},
  number={2},
  pages={125--134},
  year={2003},
  publisher={IEEE}
}

@article{nakahira2012minimum,
  title={Minimum-Bayes-cost discrimination for symmetric quantum states},
  author={Nakahira, Kenji and Usuda, Tsuyoshi Sasaki},
  journal={Physical Review A—Atomic, Molecular, and Optical Physics},
  volume={86},
  number={6},
  pages={062305},
  year={2012},
  publisher={APS}
}

@article{wieczorek2018entropic,
  title={Entropic upper bound for Bayes risk in the quantum case},
  author={Wieczorek, Rafa{\l} and Podsedkowska, Hanna},
  journal={Probability and Mathematical Statistics},
  volume={38},
  number={2},
  pages={429},
  year={2018}
}

@article{diamond2016cvxpy,
  author  = {Steven Diamond and Stephen Boyd},
  title   = {{CVXPY}: {A} {P}ython-embedded modeling language for convex optimization},
  journal = {Journal of Machine Learning Research},
  year    = {2016},
  volume  = {17},
  number  = {83},
  pages   = {1--5},
}

@article{agrawal2018rewriting,
  author  = {Agrawal, Akshay and Verschueren, Robin and Diamond, Steven and Boyd, Stephen},
  title   = {A rewriting system for convex optimization problems},
  journal = {Journal of Control and Decision},
  year    = {2018},
  volume  = {5},
  number  = {1},
  pages   = {42--60},
}

@article{holevo1982testing,
  title={Testing statistical hypotheses in quantum theory},
  author={Holevo, AS},
  journal={Probab. Math. Stat},
  volume={3},
  number={113},
  pages={29},
  year={1982}
}

@article{brody1996bayesian,
  title={Bayesian inference in quantum systems},
  author={Brody, Dorje C and Meister, Bernhard},
  journal={Physica A: Statistical Mechanics and its Applications},
  volume={223},
  number={3-4},
  pages={348--364},
  year={1996},
  publisher={Elsevier}
}

@article{rudolph2003unambiguous,
  title={Unambiguous discrimination of mixed states},
  author={Rudolph, Terry and Spekkens, Robert W and Turner, Peter S},
  journal={Physical Review A},
  volume={68},
  number={1},
  pages={010301},
  year={2003},
  publisher={APS}
}

@article{eldar2003mixed,
  title={Mixed-quantum-state detection with inconclusive results},
  author={Eldar, Yonina C},
  journal={Physical Review A},
  volume={67},
  number={4},
  pages={042309},
  year={2003},
  publisher={APS}
}

@article{touzel2007optimal,
  title={Optimal bounded-error strategies for projective measurements in nonorthogonal-state discrimination},
  author={Touzel, Maximilian Puelma and Adamson, Robert B A and Steinberg, Aephraim M},
  journal={Physical Review A},
  volume={76},
  number={6},
  pages={062314},
  year={2007},
  publisher={APS}
}

@article{tamaki2003security,
  title={Security of the {B}ennett 1992 quantum-key distribution protocol against individual attack over a realistic channel},
  author={Tamaki, Kiyoshi and Koashi, Masato and Imoto, Nobuyuki},
  journal={Physical Review A},
  volume={67},
  number={3},
  pages={032310},
  year={2003},
  publisher={APS}
}

@article{bae2013structure,
  title={Structure of minimum-error quantum state discrimination},
  author={Bae, Joonwoo},
  journal={New Journal of Physics},
  volume={15},
  number={7},
  pages={073037},
  year={2013},
  publisher={IOP Publishing}
}

@article{dalla2015optimality,
  title={Optimality of square-root measurements in quantum state discrimination},
  author={Dalla Pozza, Nicola and Pierobon, Gianfranco},
  journal={Physical Review A},
  volume={91},
  number={4},
  pages={042334},
  year={2015},
  publisher={APS}
}

@article{chefles2000quantum,
  title={Quantum state discrimination},
  author={Chefles, Anthony},
  journal={Contemporary Physics},
  volume={41},
  number={6},
  pages={401--424},
  year={2000},
  publisher={Taylor \& Francis}
}

@article{bergou200411,
  title={{D}iscrimination of quantum states},
  author={Bergou, J{\'a}nos A and Herzog, Ulrike and Hillery, Mark},
  journal={Quantum State Estimation},
  pages={417--465},
  year={2004},
  publisher={Springer}
}

@inproceedings{bergou2007quantum,
  title={Quantum state discrimination and selected applications},
  author={Bergou, J{\'a}nos A},
  booktitle={Journal of Physics: Conference Series},
  volume={84},
  number={1},
  pages={012001},
  year={2007},
  organization={IOP Publishing}
}

@article{barnett2009quantum,
  title={Quantum state discrimination},
  author={Barnett, Stephen M and Croke, Sarah},
  journal={Advances in Optics and Photonics},
  volume={1},
  number={2},
  pages={238--278},
  year={2009},
  publisher={Optica Publishing Group}
}

@article{bergou2010discrimination,
  title={Discrimination of quantum states},
  author={Bergou, J{\'a}nos A},
  journal={Journal of Modern Optics},
  volume={57},
  number={3},
  pages={160--180},
  year={2010},
  publisher={Taylor \& Francis}
}

@article{strubi2013measuring,
  title={Measuring ultrasmall time delays of light by joint weak measurements},
  author={Str{\"u}bi, Gr{\'e}gory and Bruder, C},
  journal={Physical Review Letters},
  volume={110},
  number={8},
  pages={083605},
  year={2013},
  publisher={APS}
}

@article{eldar2004optimal,
  title={Optimal detection of symmetric mixed quantum states},
  author={Eldar, Yonina C and Megretski, Alexandre and Verghese, George C},
  journal={IEEE Transactions on Information Theory},
  volume={50},
  number={6},
  pages={1198--1207},
  year={2004},
  publisher={IEEE}
}

@article{andersson2002minimum,
  title={Minimum-error discrimination between three mirror-symmetric states},
  author={Andersson, Erika and Barnett, Stephen M and Gilson, Claire R and Hunter, Kieran},
  journal={Physical Review A},
  volume={65},
  number={5},
  pages={052308},
  year={2002},
  publisher={APS}
}

@article{jaeger1995optimal,
  title={Optimal distinction between two non-orthogonal quantum states},
  author={Jaeger, Gregg and Shimony, Abner},
  journal={Physics Letters A},
  volume={197},
  number={2},
  pages={83--87},
  year={1995},
  publisher={Elsevier}
}

@article{combes2015cost,
  title={Cost of postselection in decision theory},
  author={Combes, Joshua and Ferrie, Christopher},
  journal={Physical Review A},
  volume={92},
  number={2},
  pages={022117},
  year={2015},
  publisher={APS}
}

@article{bandyopadhyay2014conclusive,
  title={Conclusive exclusion of quantum states},
  author={Bandyopadhyay, Somshubhro and Jain, Rahul and Oppenheim, Jonathan and Perry, Christopher},
  journal={Physical Review A},
  volume={89},
  number={2},
  pages={022336},
  year={2014},
  publisher={APS}
}

@article{aminikhanghahi2017survey,
  title={A survey of methods for time series change point detection},
  author={Aminikhanghahi, Samaneh and Cook, Diane J},
  journal={Knowledge and information systems},
  volume={51},
  number={2},
  pages={339--367},
  year={2017},
  publisher={Springer}
}

@article{mitarai2019methodology,
  title={Methodology for replacing indirect measurements with direct measurements},
  author={Mitarai, Kosuke and Fujii, Keisuke},
  journal={Physical Review Research},
  volume={1},
  number={1},
  pages={013006},
  year={2019},
  publisher={APS}
}

@article{bharti2021iterative,
  title={Iterative quantum-assisted eigensolver},
  author={Bharti, Kishor and Haug, Tobias},
  journal={Physical Review A},
  volume={104},
  number={5},
  pages={L050401},
  year={2021},
  publisher={APS}
}

@article{ko2018advanced,
  title={Advanced unambiguous state discrimination attack and countermeasure strategy in a practical {B}92 {QKD} system},
  author={Ko, Heasin and Choi, Byung-Seok and Choe, Joong-Seon and Youn, Chun Ju},
  journal={Quantum Information Processing},
  volume={17},
  pages={1--14},
  year={2018},
  publisher={Springer}
}

@article{hayashi2008quantum,
  title={Quantum measurements for hidden subgroup problems with optimal sample complexity},
  author={Hayashi, M and Kawachi, A and Kobayashi, H},
  journal={Quantum Information and Computation},
  volume={8},
  number={3\&4},
  pages={345--358},
  year={2008},
  publisher={Rinton Press}
}

@article{buhrman2001quantum,
  title={Quantum fingerprinting},
  author={Buhrman, Harry and Cleve, Richard and Watrous, John and De Wolf, Ronald},
  journal={Physical Review Letters},
  volume={87},
  number={16},
  pages={167902},
  year={2001},
  publisher={APS}
}

@book{watrous2018theory,
  title={The theory of quantum information},
  author={Watrous, John},
  year={2018},
  publisher={Cambridge university press}
}

@article{huang2019near,
  title={Near-term quantum algorithms for linear systems of equations},
  author={Huang, Hsin-Yuan and Bharti, Kishor and Rebentrost, Patrick},
  journal={arXiv preprint arXiv:1909.07344},
  year={2019}
}

@misc{aharonov2006polynomial,
      title={A Polynomial Quantum Algorithm for Approximating the {J}ones Polynomial}, 
      author={Dorit Aharonov and Vaughan Jones and Zeph Landau},
      year={2006},
      eprint={quant-ph/0511096},
      archivePrefix={arXiv},
      primaryClass={quant-ph}
}

@inproceedings{gilyen2019quantum,
  title={Quantum singular value transformation and beyond: exponential improvements for quantum matrix arithmetics},
  author={Gily{\'e}n, Andr{\'a}s and Su, Yuan and Low, Guang Hao and Wiebe, Nathan},
  booktitle={Proceedings of the 51st Annual ACM SIGACT Symposium on Theory of Computing},
  pages={193--204},
  year={2019}
}

@misc{cvx,
  author       = {{CVX Research, Inc.}},
  title        = {{CVX}: Matlab Software for Disciplined Convex Programming, version 2.0},
  howpublished = {\url{http://cvxr.com/cvx}},
  month        = aug,
  year         = 2012
}

@article{uola2020all,
  title={All quantum resources provide an advantage in exclusion tasks},
  author={Uola, Roope and Bullock, Tom and Kraft, Tristan and Pellonp{\"a}{\"a}, Juha-Pekka and Brunner, Nicolas},
  journal={Physical Review Letters},
  volume={125},
  number={11},
  pages={110402},
  year={2020},
  publisher={APS}
}

@article{johnston2023tight,
  title={Tight bounds for antidistinguishability and circulant sets of pure quantum states},
  author={Johnston, Nathaniel and Russo, Vincent and Sikora, Jamie},
  journal={arXiv preprint arXiv:2311.17047},
  year={2023}
}

@inproceedings{aharonov2000quantum,
  title={Quantum bit escrow},
  author={Aharonov, Dorit and Ta-Shma, Amnon and Vazirani, Umesh V and Yao, Andrew C},
  booktitle={Proceedings of the thirty-second annual ACM symposium on Theory of computing},
  pages={705--714},
  year={2000}
}

@inproceedings{ambainis2001new,
  title={A new protocol and lower bounds for quantum coin flipping},
  author={Ambainis, Andris},
  booktitle={Proceedings of the thirty-third annual ACM symposium on Theory of computing},
  pages={134--142},
  year={2001}
}

@article{sikora2017simple,
  title={Simple, near-optimal quantum protocols for die-rolling},
  author={Sikora, Jamie},
  journal={Cryptography},
  volume={1},
  number={2},
  pages={11},
  year={2017},
  publisher={MDPI}
}

@article{mishra2023optimal,
  title={On the optimal error exponents for classical and quantum antidistinguishability},
  author={Mishra, Hemant K and Nussbaum, Michael and Wilde, Mark M},
  journal={arXiv preprint arXiv:2309.03723},
  year={2023}
}

@article{darkhovsky2000non,
  title={Non-Parametric Statistical Diagnosis},
  author={Darkhovsky, BS},
  year={2000},
  publisher={Kluwer Academic Publishers}
}

@misc{johnston2025complexityperfectquantumstate,
      title={The complexity of perfect quantum state classification}, 
      author={Nathaniel Johnston and Benjamin Lovitz and Vincent Russo and Jamie Sikora},
      year={2025},
      eprint={2510.20789},
      archivePrefix={arXiv},
      primaryClass={quant-ph},
      url={https://arxiv.org/abs/2510.20789}, 
}

\appendix

\section{Background}\label{sec:bg}

In this section, we introduce some background material relevant to this work.

\subsection{Mathematical background} 
       
The \emph{Gram matrix} $G$ of a set of vectors $\{v_1, \dots, v_N\}$ is a Hermitian positive semidefinite matrix whose entries are given by $G_{ij} = \hsip{v_i}{v_j}$.
We use the notation $A \succeq B$ meaning that $A - B$ is positive semidefinite and $A \succ B$ meaning that $A - B$ is positive definite. 

\bigskip

An $N \times N$ \emph{Toeplitz matrix} $T$ is a matrix with entries satisfying  
\begin{equation}
    T_{i, j} = T_{i + 1, j + 1} \text{ for all } i, j = \{1, \dots, N-1\}.
\end{equation}

This Toeplitz matrix can be written as
\begin{align*}
    T = 
        \begin{pmatrix}
            t_0 & t_{-1} & t_{-2} & \cdots & t_{-(N - 1)} \\
            t_1 & t_0 & t_{-1} & \cdots & t_{-(N - 2)} \\
            t_2 & t_1 & t_0 & \cdots & t_{-(N - 3)} \\
            \vdots & \vdots & \vdots & \ddots & \vdots \\
            t_{N - 1} & t_{N - 2} & t_{N - 3} & \cdots & t_0
        \end{pmatrix},
\end{align*}
noting the constant superdiagonals.
Since each diagonal of the Toeplitz matrix has the same value, we can alternately express $T$ as
\begin{align}
    \label{eq:Toep_mat}
    T = \sum_{k = -(N-1)}^{N-1} t_k \Theta_{k},
\end{align}
where $\Theta_k$ is an $N \times N$ Toeplitz matrix with ones on the $k$-th diagonal and zeros elsewhere.
Here, $k = 0$ indicates the principal diagonal, $k > 0$ denotes the sub-diagonals, and $k < 0$ corresponds to the super-diagonals.
Note that a Toeplitz matrix is not necessarily square. 

\emph{Semidefinite programming} is an area of convex optimization where the goal is to optimize a linear function of a positive semidefinite matrix $X$ over affine constraints. 
A semidefinite program (abbreviated as SDP) can be written in standard form
    \begin{equation*}    
        \begin{aligned}
            \alpha = \min \quad & \hsip{C}{X} \\
            \st \quad & \hsip{A_i}{X} = b_i,\ i \in \{1, \dots, m\} \\
            & X \succcurlyeq 0,
        \end{aligned}
    \end{equation*}
where the matrices $C$ and $A_i$ are Hermitian and $b_i$ are real. 
Every SDP has a dual which is an SDP itself and is defined as 
    \begin{equation}
        \label{eq:power_rangers_SPD}
        \begin{aligned}
            \beta = \max \quad & \hsip{b}{y} \\
            \st \quad & \sum_{i=1}^m y_i A_i \preccurlyeq C \\
            & y \in \mathbb{R}^m.  
        \end{aligned}
    \end{equation}     
We have that $\alpha \leq \beta$, a fact known as weak duality. 
If both $\alpha$ and $\beta$ are finite and one has a feasible solution where the inequalities are strict (known as strict feasibility), then $\alpha = \beta$. 
This condition is known as strong duality. 


\section{Application: The quantum changepoint identification problem}\label{sec:mQCP}  

We now consider a general case of the quantum changepoint problem.
Suppose again that Alice promised to deliver $N$ copies of the state $\ket{\psi}$ to Bob.
At a point $c_1$, Alice starts to generate a mutated state $\ket{\phi_1}$.
Then at some point $c_2$, Alice begins to generate another mutated state $\ket{\phi_2}$.
Let us assume that there are $P$ such points in total leading to the states $\{\ket{\phi_3}, \dots, \ket{\phi_P}\}$ at points $\{c_3, \dots, c_P\}$ respectively.
Therefore, Bob must now consider the states
\begin{align}
    \label{eq:CP_P}
    \ket{\tau_{c_1, \cdots, c_P}} = \ket{\psi}^{\otimes c_1}\ \otimes\ \ket{\phi_1}^{\otimes (c_2 - c_1)}\ \otimes\ \cdots\ \otimes\ \ket{\phi_{P-1}}^{\otimes (c_P - c_{P-1})}\ \otimes\ \ket{\phi_P}^{\otimes (N - c_P)}.
\end{align}

The primal-dual pair in Eq.~(\ref{prog:P_D}) can be expressed as 
\begin{equation}
\begin{minipage}{.4\textwidth}
    \begin{equation*}
        \begin{aligned}
            \alpha = \max \quad & \sum_{g\ \in\ \Iset_L}\ \sum_{c\ \in\ \Iset_N} R_{g, c}\ q_c\ \bra{\tau_c} M_g \ket{\tau_c} \\
            \st \quad & \sum_{g\ \in\ \Iset_L} M_g = \id \\
            & \hspace{35pt} M_g \succcurlyeq 0,\ g \in \Iset_L, 
        \end{aligned}
    \end{equation*}    
\end{minipage} 
\begin{minipage}{.55\textwidth} 
    \begin{equation*}
        \label{prog:D_mQCP}
        \begin{aligned}
            \beta = \min \quad & \Tr(Y) \\
            \st \quad & Y \succcurlyeq \sum_{c\ \in\ \Iset_N} R_{g, c}\ q_c\ \ketbra{\tau_c},\ g \in \Iset_L, \\
            & \phantom{N}
        \end{aligned}
    \end{equation*}
\end{minipage}
\end{equation} 
where
\begin{equation} 
\Iset_m := \{(a_1, \cdots, a_P)\ :\ 1 \le a_1 \le \cdots \le a_P \le m\ \text{ where }\ a_i < a_{i+1}\ \text{ unless }\ a_i = m\}.
\end{equation} 
This indexing means that no state $\ket{\phi_i}$ is skipped, unless $N$ states total have already been sent. In other words, you cannot skip from $\ket{\phi_i}$ to $\ket{\phi_{i+2}}$, say, but if the $N$-th state is $\ket{\psi_i}$, then you will not see any copies of $\ket{\psi_j}$ for $j > i$.

The corresponding reduced primal-dual pair becomes
\begin{equation}
\begin{minipage}{.4\textwidth}
    \begin{equation*}
        \begin{aligned}
            \alpha' = \max \quad & \sum_{g\ \in\ \Iset_L}\ \sum_{c\ \in\ \Iset_N} R_{g, c}\ q_c\ \bra{c} W_g \ket{c} \\
            \st \quad & \sum_{g\ \in\ \Iset_L} W_g = G \\
            & \hspace{35pt} W_g \succcurlyeq 0,\ g \in \Iset_L.
        \end{aligned}
    \end{equation*}    
\end{minipage} 
\begin{minipage}{.55\textwidth} 
    \begin{equation*}
        \label{prog:D'_mQCP}
        \begin{aligned}
            \beta' = \min \quad & \hsip{X}{G} \\
            \st \quad & X \succcurlyeq \sum_{c\ \in\ \Iset_N} R_{g, c}\ q_c\ \ketbra{c},\ g \in \Iset_L.\\
            & \phantom{N}
        \end{aligned}
    \end{equation*}
\end{minipage}
\end{equation}

Let us denote the overlaps as follows
\begin{equation}
    \label{eq:inner_prod}
    \begin{aligned}
        |\braket{\psi}{\phi_i}| = \gamma_{i},\ \quad |\braket{\phi_i}{\phi_j}| = |\braket{\phi_j}{\phi_i}| = \gamma_{ij},
        \quad i \ne j \text{ and } i,\ j \in \{1, \dots, P\}.
    \end{aligned}
\end{equation}

Each of these overlaps can be computed using the swap test (see Section~\ref{sec:QC_bg}).

\subsection{At most one changepoint}\label{sec:1CP}

When there is at most one changepoint, Bob discriminates between the following states
\begin{equation}
    \begin{aligned}
        \ket{\tau_1} &\coloneqq \ket{\tau_{1NN\cdots NN}} = \ket{\psi}\ \otimes\ \ket{\phi_1}^{\otimes N-1} \\
        \ket{\tau_2} &\coloneqq \ket{\tau_{2NN\cdots NN}} = \ket{\psi}^{\otimes 2}\ \otimes\ \ket{\phi_1}^{\otimes N-2} \\
        \vdots \\
        \ket{\tau_{N-1}} &\coloneqq \ket{\tau_{N-1,NN\cdots NN}} = \ket{\psi}^{\otimes N-1}\ \otimes\ \ket{\phi_1} \\
        \ket{\tau_N} &\coloneqq \ket{\tau_{NNN\cdots NN}} = \ket{\psi}^{\otimes N}. 
    \end{aligned}
\end{equation}
Notice that the last state $\ket{\tau_N}$ indicates the sequence where no changepoint has occurred, i.e, Alice generated all the states as promised.
All the remaining states correspond to the sequences where exactly one changepoint occurred.
Therefore, Bob wishes to not only determine whether a changepoint occurred, but also the exact point of said occurrence.

Here we assume that the changepoint follows a uniform prior, i.e, each of the states $\ket{\tau_k}$, \\ $k \in \{1, \dots, N\}$, are equally likely.

\begin{lemma}
    \label{lem:G_sym_toep}
    Gram matrix $G$ of the possible sequences generated by Alice when there is at most one changepoint is a Hermitian Toeplitz matrix.
\end{lemma}
\begin{proof}
    By direct computation, we have
    \begin{align*}
        G_{ij} = 
        \begin{cases}
            \braket{\phi_1}{\psi}^{(j-i)} &i < j \\
            \braket{\psi}{\phi_1}^{(i-j)} &\text{otherwise}
        \end{cases}
    \end{align*}
    $G$ is clearly a Hermitian Toeplitz matrix from the above.
\end{proof}

\begin{lemma}
    \label{lem:T_sym_Toep}
    When there is at most one changepoint, the Gram matrix $G$ has the same reward as the Gram matrix $T^{(1)}$ where $(T^{(1)})_{ij} = |G_{ij}|$. Moreover, $T^{(1)}$ is a symmetric Toeplitz matrix.
\end{lemma}
\begin{proof}
    Define $\ket{\phi_1'} = e^{i \theta} \ket{\phi_1}$, where we choose $\theta \in \reals$ such that $\braket{\phi_1'}{\psi} = |\braket{\phi_1}{\psi}|$.
    From~\eqref{eq:gen_QSD2}, the value of the SDP clearly remains same if we replace any state with one that is equivalent up to a global phase.
    Define $\ket{\tau_i'} = \ket{\psi}^{\otimes i}\ \otimes\ \ket{\phi_1'}^{\otimes (N-i)}$, for all values of $i$.
    Then $\ket{\tau_i}$ and $\ket{\tau_i'}$ differ by a global phase, thus the SDPs have the same value.
    Let $T^{(1)}$ by the Gram matrix of $\ket{\tau_1'}, \dots, \ket{\tau_N'}$.
    Notice that $(T^{(1)})_{ij} = |\braket{\psi}{\phi_1}|^{|i-j|} = |G_{ij}|$. 
    Finally, $T^{(1)}$ is symmetric Toeplitz, as desired. 
\end{proof} 

The dual SDP in Eq.~(\ref{prog:D'_mQCP}) can be expressed as
\begin{equation}
    \label{eq:1CP_Toep}
    \begin{aligned}
        \beta'_{T1} = \min \quad& \hsip{X}{T^{(1)}} \\
        \st \quad& X \succcurlyeq \frac{1}{N} \sum_{j=1}^N R_{ij} | j \rangle\langle j |, i \in \{1, \dots, N\} \\
    \end{aligned}
\end{equation}

Consider the following reward scheme
\begin{align}
    \label{eq:R_1CP}
    (R_{1CP})_{ij} = 
    \begin{cases}
        r_{|i - j|},\ i \in \{1, \dots, N\} \\
        c,\ \hspace{15pt}\ i = N + 1
    \end{cases}
\end{align}
for $j \in \{1, \dots, N\}$, where $r_0, \dots, r_{N-1}$ and $c$ are real scalars.
The motivation behind this strategy is as follows.
When Bob is able to correctly identify the changepoint, he is awarded a reward of $r_0$, regardless of when the changepoint occurred.
Similarly, if the changepoint occurred at the $j$-th state, and Bob guessed the $i$-th state, irrespective of whether $i < j$ or $i > j$, he is given a reward of $r_{|i-j|}$.
Notwithstanding the state Bob received, if he gives an inconclusive outcome, he is awarded a constant reward of $c$.

Accounting for all this, we can rewrite the dual SDP in Eq.~(\ref{prog:D'_mQCP}) as follows
\begin{equation}
    \label{prog:D'_1CP}
    \begin{aligned}
        \beta'_{1CP} = \min \quad & \hsip{X_{1CP}}{T_{1CP}} \\
        \st \quad & X_{1CP} \succcurlyeq \frac{1}{N}\, \sum_{j=1}^N r_{|i - j|} \ketbra{j},\ i \in \{1, \dots, N\} \\
        & X_{1CP} \succcurlyeq \frac{c}{N}\, \id_N.
    \end{aligned}
\end{equation}

\paragraph{A heuristic.}
Observe that $(R_{1CP})_{ij} = r_{|i-j|},\ i, j \in \{1, \dots, N\}$ corresponds to a symmetric Toeplitz matrix.
Since $T_{1CP}$ is a symmetric Toeplitz matrix and $R_{1CP}$ has a Toeplitz-like structure, we investigate a heuristic solution $X'_{1CP}$ which is restricted to be symmetric Toeplitz.
The incentive behind such a restriction is that a symmetric Toeplitz matrix can be fully constructed from its first row, and thus has fewer parameters.
Then both $T_{1CP}$ and $X'_{1CP}$ can be expressed as
\begin{align}
    \label{eq:sym_Toep}
    T_{1CP} = \sum_{k=-(N-1)}^{(N-1)} \gamma^{|k|} \Theta_k \quad\qquad\quad X'_{1CP} = \sum_{k=-(N-1)}^{(N-1)} x_{|k|} \Theta_k,
\end{align}
where $\gamma = |\braket{\psi}{\phi_1}|$, and $\Theta_k$ is a $N \times N$ Toeplitz matrix with ones on the $k$-th diagonal and zeros elsewhere.
Here, $k = 0$ indicates the principal diagonal, $k > 0$ denotes the sub-diagonals, and $k < 0$ corresponds to the super-diagonals.

We can express Eq.~(\ref{prog:D'_1CP}) with this heuristic as
\begin{equation}
    \label{prog:D''_1CP}
    \begin{aligned}
        \beta''_{1CP} = \min \quad & \hsip{x}{p} \\
        \st \quad & \sum_{k=0}^{N-1} x_k \Theta_k \succcurlyeq \frac{1}{N}\, \sum_{j=1}^N r_{|i - j|} \ketbra{j},\ i \in \{1, \dots, N\} \\
        & \sum_{k=0}^{N-1} x_k \Theta_k \succcurlyeq \frac{c}{N}\, \id_N,
    \end{aligned}
\end{equation}
where $p = \left( N,\ 2(N-1) \gamma,\ 2(N-2) \gamma^2,\ \ldots,\ 2 \gamma^{N-1} \right)^\top$ and $x = \left( x_0, x_1, x_2, \ldots, x_{N-1} \right)^\top$.

In Appendix~\ref{sec:expts} we examine the cases when there are two or more changepoints.
The calculations are straightforward but somewhat tedious.

\paragraph{Related work.} 
For the minimum-error case,~\cite{sentis2016quantum} showed that in the limit of the sequence length $N$, the probability of success is bounded by
\begin{align}
    \label{eq:sentis2016bound}
    P_D \leq \frac{4 \sqrt{1 - \gamma^2}}{\pi} K^2(\gamma^2) + 4 \left( \frac{1 + \gamma}{1 - \gamma} \right)^{3/2} \frac{1}{N^{1 - \varepsilon}}
\end{align}
where $K(x)$ is the complete elliptic function of the first kind~\cite{darkhovsky2000non}, and $\varepsilon > 0$ is an arbitrary constant.
For a sequence of length $N = 50$, Figure~\ref{fig:ME_compare} compares the difference between the heuristic (Eq.~\eqref{prog:D''_1CP}) and the optimal SDP (Eq.~\eqref{prog:D'_1CP} where $r_{|i-j|} = \delta_{ij}$ for all $i, j$, and $c = 0$), with the difference between Eq.~\eqref{eq:sentis2016bound} and the aforementioned optimal SDP, for varying values of $\gamma^2$.
Observe that while our heuristic outperforms the bound in Eq.~\eqref{eq:sentis2016bound}, the heuristic requires us to solve an SDP while Eq.~\eqref{eq:sentis2016bound} provides a closed-form expression.
\begin{figure}[h!]
    \centering
    \includegraphics[scale=0.6]{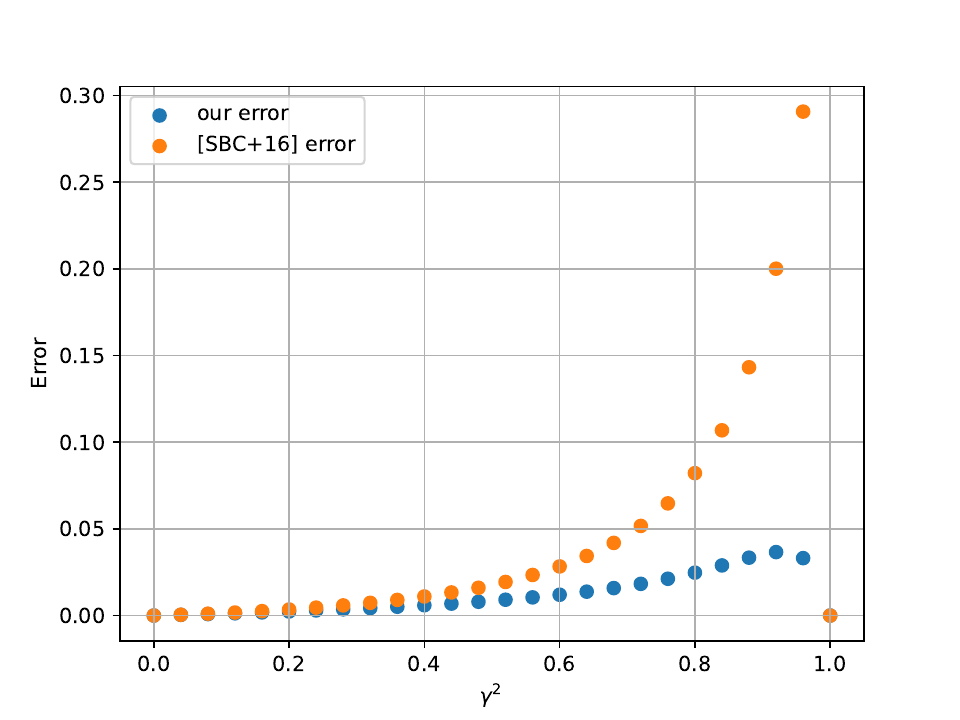}
    \caption{
    For a sequence of length $N = 50$, the difference between the heuristic (Eq.~\eqref{prog:D''_1CP}) and the optimal SDP is compared with the difference between Eq.~\eqref{eq:sentis2016bound} and the optimal SDP, as we vary the value of $\gamma^2 = |\bra{\psi}\ket{\phi_1}|^2$.
    Although our heuristic is a better bound than Eq.~\eqref{eq:sentis2016bound}, the heuristic requires us to solve an SDP while Eq.~\eqref{eq:sentis2016bound} is a closed-form expression.
    }
    \label{fig:ME_compare}
\end{figure} 

Similarly for the unambiguous case,~\cite{sentis2017exact} found that the success probability $P_u$ can be approximated by the expression
\begin{align}
    \label{eq:sentis2017approx}
    P_{u} \approx
    \begin{cases}
        \frac{1 - \gamma}{1 + \gamma} + \frac{1}{N} \frac{2 \gamma}{(1 + \gamma)^2}, \qquad\qquad\qquad\qquad 0 \leq \gamma \leq \gamma^* = (\sqrt{5} - 1)/2 \\
        \frac{1 - \gamma}{1 + \gamma} + \frac{1}{N} \frac{2 \gamma}{(1 + \gamma)^2} - \frac{2}{N} \left( \frac{1 - \gamma - \gamma^2}{1 + \gamma} \right)^2, \quad \gamma* < \gamma \leq 1.
    \end{cases}
\end{align}
For a sequence of length $N = 15$, Figure~\ref{fig:UA_compare} compares the difference between the heuristic (Eq.~\eqref{prog:D''_1CP}) and the optimal SDP (Eq.~\eqref{prog:D'_1CP} where $r_{|i-j|} = 1$ if $i = j$, $r_{|i-j|} = -M$ if $i \ne j$, and $c = 0$, for $M = 10^4$, with the difference between Eq.~\eqref{eq:sentis2017approx} and this optimal SDP, for varying values of $\gamma$.
For $\gamma \le \gamma^* = (\sqrt{5} - 1)/2$, the error between the heuristic and this optimal SDP is of the order $10^{-5}$.
When $\gamma > \gamma^*$, observe that Eq.~\eqref{eq:sentis2017approx} provides a much better bound compared to the heuristic. 

\begin{remark}
We remark that our original SDP value is an approximation on unambiguous discrimination since we cannot numerically set $M = \infty$. 
Thus, Figure~\ref{fig:UA_compare} compares the two heuristics compared to this SDP approximation. 
We see that in this case, both heuristics perform well. 
\end{remark}

\begin{figure}[h!]
    \centering
    \includegraphics[scale=0.6]{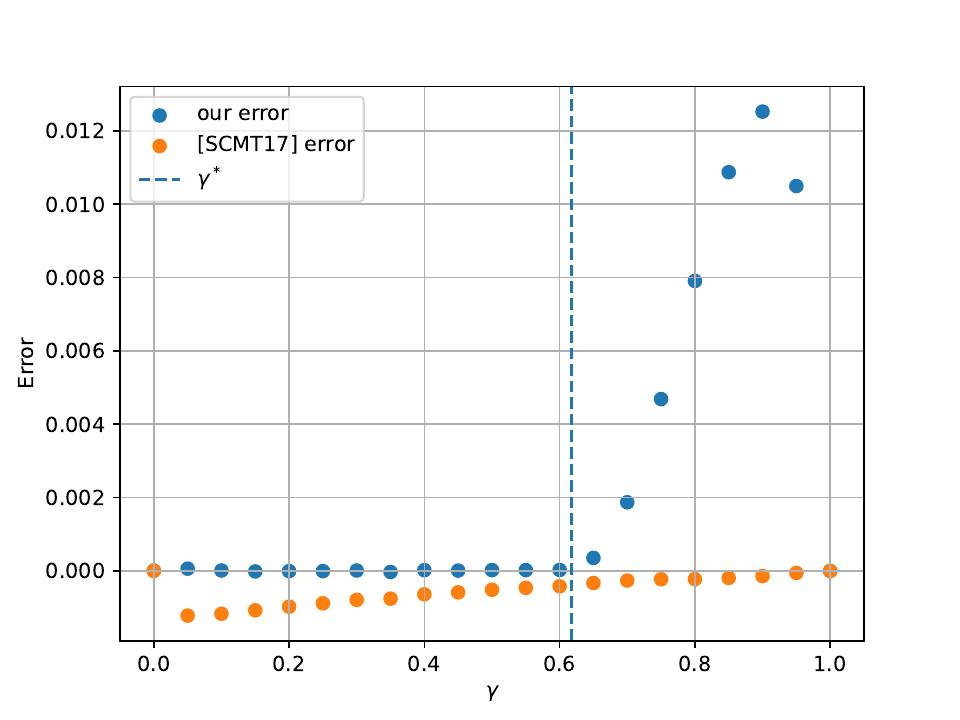} 
    \caption{
    For a sequence of length $N = 15$, we compare the difference between the heuristic (Eq.~\eqref{prog:D''_1CP}) and the optimal SDP, with the difference between Eq.~\eqref{eq:sentis2016bound} and the optimal SDP, for varying values of $\gamma$.
    For $\gamma \le \gamma^* = (\sqrt{5} - 1)/2$, the error between the heuristic and the optimal SDP is of the order $10^{-5}$.
    When $\gamma > \gamma^*$, observe that Eq.~\eqref{eq:sentis2016bound} provides a much better bound compared to the heuristic.
    }
    \label{fig:UA_compare}
\end{figure}

\section{Numerical experiments}\label{sec:expts}

To demonstrate the performance of our heuristic approach, we consider the following states
\begin{align}
    \label{eq:theta_states}
    \ket{\psi} = \ket{0}, \quad \ket{\phi_k} = \cos(k \theta) \ket{0} + \sin(k \theta) \ket{1},\ k \in \{1, \dots, P\}.
\end{align} 

Consider the case where Alice promises Bob $N$ copies of the state $\ket{0}$.
Bob suspects the following:
(1) At some point $c_1$, her device might have mutated from generating the state $\ket{0}$ to the state $\ket{+}$.
(2) At a different point $c_2$, a further mutation might have caused the switch from the state $\ket{+}$ to the state $\ket{1}$.
(3) At a third point $c_3$, one last mutation could have led to $\ket{1}$ turning it to $-\ket{-}$.
The possible mutations are as described in Figure~\ref{fig:pi_4_states}.

\def\xmax{2.0}
\def\ul{0.6}
\def\R{2.0}
\def\ang{45}

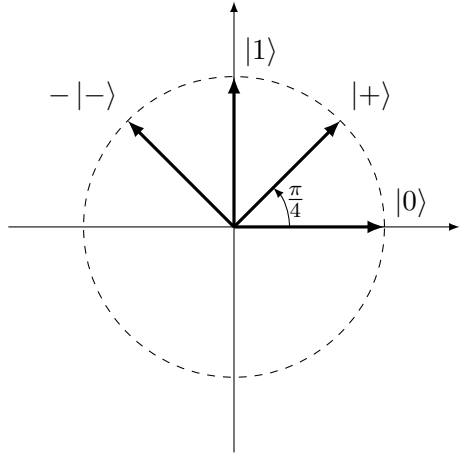
\begin{figure}[h]
    \centering
    \begin{tikzpicture}
        \coordinate (O) at (0,0);
        \coordinate (zero) at (\xmax,0);
        \coordinate (plus) at (\ang:\R);
        \coordinate (one) at (0,\xmax);
        \coordinate (neg_minus) at (3*\ang:\R);
  
        \draw[->] (-1.5*\xmax,0) -- (1.5*\xmax,0);
        \draw[->] (0,-1.5*\xmax) -- (0,1.5*\xmax);
  
        \node[above right] at (zero) {$\ket{0}$};
        \draw[vector] (O) -- (zero);
        \node[above right] at (plus) {$\ket{+}$};
        \draw[vector] (O) -- (plus);
        \node[above right] at (one) {$\ket{1}$};
        \draw[vector] (O) -- (one);
        \node[above left] at (neg_minus) {$-\ket{-}$};
        \draw[vector] (O) -- (neg_minus);

        \draw pic[->,"$\frac{\pi}{4}$",draw=black,angle radius=21,angle eccentricity=1.2] {angle=zero--O--plus};

        \draw[dashed] (O) circle (\R);
    \end{tikzpicture}
    \caption{The state $\ket{0}$ promised by Alice and the possible mutated states $\ket{+}, \ket{1}$ and $-\ket{-}$ that Bob suspects of receiving.}
    \label{fig:pi_4_states}
\end{figure}

Upon receiving all of the $N$ states, Bob can exploit the Bayes approach by considering the following states
\begin{equation}
    \label{eq:Bob's_suspicion_succinct}
    \begin{aligned}
        \ket{\tau_{c_1 c_2 c_3}} = -\ket{0}^{\otimes c_1}\ \otimes \ket{+}^{\otimes (c_2 - c_1)}\ \otimes \ket{1}^{\otimes (c_3 - c_2)}\ \otimes \ket{-}^{\otimes (N - c_3)} \\
        (c_1, c_2, c_3) \in \Iset_N
    \end{aligned}
\end{equation}

When no changepoints occur, Bob is looking for the state
\begin{align}
    \label{eq:no_CP_eg}
    \ket{\tau_{N N N}} = \ket{0}^{\otimes N}.
\end{align}
If he suspects a mutation to the state $\ket{+}$, he looks for
\begin{align}
    \label{eq:1CP_eg}
    \ket{\tau_{c_1 N N}} = \ket{0}^{\otimes c_1}\ \otimes \ket{+}^{\otimes (N - c_1)}.
\end{align}
A suspicion of a further mutation to state $\ket{1}$ requires looking for the state
\begin{align}
    \label{eq:2CP_eg}
    \ket{\tau_{c_1 c_2 N}} = \ket{0}^{\otimes c_1}\ \otimes \ket{+}^{\otimes (c_2 - c_1)}\ \otimes \ket{1}^{\otimes (N - c_2)}.
\end{align}
When Bob suspects that all of the possible mutations may have occurred, he looks for
\begin{align}
    \label{eq:3CP_eg}
    \ket{\tau_{c_1 c_2 c_3}} = - \ket{0}^{\otimes c_1}\ \otimes \ket{+}^{\otimes (c_2 - c_1)}\ \otimes \ket{1}^{\otimes (c_3 - c_2)}\ \otimes \ket{-}^{\otimes (N - c_3)}.
\end{align}

\paragraph{At most one changepoint.}
Here, we look at the case when Bob suspects Alice's device of either having no mutations or ones that generate $\ket{+}$.
So Bob needs to determine exactly where this changepoint occurred.
This is accomplished by discriminating between the states described by Eq.~(\ref{eq:no_CP_eg}) and Eq.~(\ref{eq:1CP_eg}).

We consider the closer-the-better reward strategy described in eq.~\eqref{eq:R_ME_ctbr} to obtain
\begin{align}
    \label{eq:1CP_ctbr}
    R_{ij} =
    \begin{cases}
        \left( \frac{1}{\sqrt{2}} \right)^{|i-j|} & i \in \{1, \dots, N\} \\
        0 &i = N+1
    \end{cases}
\end{align}
for $j \in \{1, \dots, N\}$.

Figure~\ref{fig:beta'_beta''_pi_4_1CP} considers the solution of the reduced SDP (Eq.~\eqref{reduceddual}) and the heuristic (Eq.~\eqref{reduceddual2}).
We note the following two observations:
(1) Our reduction (Eq.~\eqref{reduceddual}) is the probability of both detecting as well as localizing a single changepoint for sequences of length $220$.
Without this reduction, solving Eq.~\eqref{eq:gen_QSD2} involves finding operators of size $2^{220} \times 2^{220}$.
(2) Observe that solving our heuristic solution (Eq.~\eqref{reduceddual2}) is roughly \emph{seven} times faster, with a gap of only about $10^{-3}$.

\begin{figure}[h]
    \centering
    \begin{subfigure}[b]{0.495\textwidth}
        \centering
        \includegraphics[width=\textwidth]{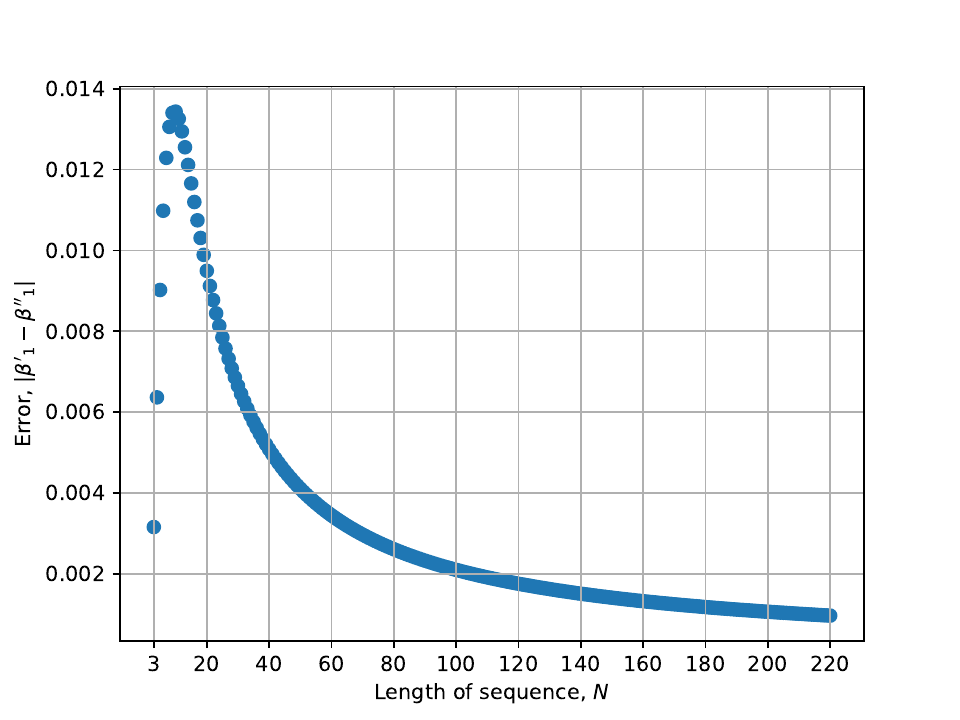}
        \caption{Absolute value of the difference between Eq.~\eqref{reduceddual} and Eq.~\eqref{reduceddual2}.}
    \end{subfigure}
    \hfill
    \begin{subfigure}[b]{0.495\textwidth}
        \centering
        \includegraphics[width=\textwidth]{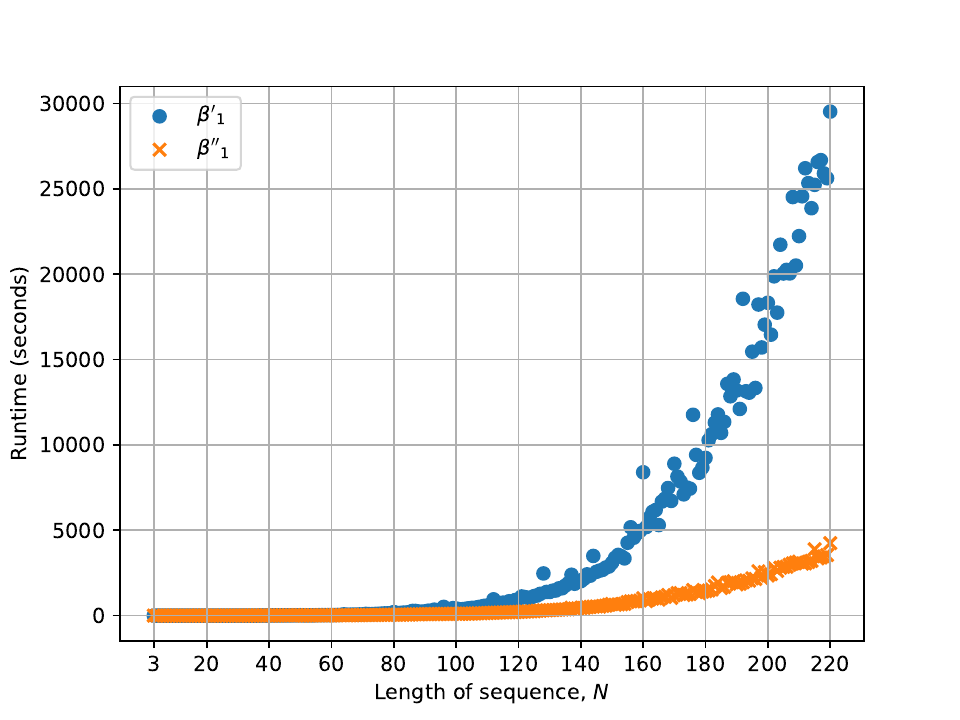}
        \caption{Runtimes of Eq.~\eqref{reduceddual} and Eq.~\eqref{reduceddual2}.}
    \end{subfigure}
    \caption{
    Here we consider the sequences where the state $\ket{0}$ mutated to the state $\ket{+}$. 
    Using the closer-the-better reward scheme with $\gamma = 1/\sqrt{2}$, we see the error between the heuristic (Eq.~\eqref{reduceddual2}) and the reduced SDP (Eq.~\eqref{reduceddual}) (on the left), and the runtimes of both Eq.~\eqref{reduceddual} and Eq.~\eqref{reduceddual2} (on the right).
    Observe that the difference between Eq.~\eqref{reduceddual} and Eq.~\eqref{reduceddual2} decreases rapidly approaching $10^{-3}$ at $N = 220$ (on the left).
    Although there seems to be no advantage, in terms of time, in computing Eq.~\eqref{reduceddual2} for sequences up to length $80$, we observe that computing Eq.~\eqref{reduceddual2} is roughly \emph{seven} times faster than calculating Eq.~\eqref{reduceddual} for larger values of $N$.
    }
    \label{fig:beta'_beta''_pi_4_1CP}
\end{figure}

\paragraph{At most two changepoints.}
In addition to the possible mutation to the state $\ket{+}$, here we look at the scenario where Bob suspects of a further mutation at some point $c_2$, to produce the state $\ket{1}$.
Bob's task is now to identify both of the changepoints, at $c_1$ as well as at $c_2$.
This can be accomplished by simply discriminating between all the possible states described in Eq.~(\ref{eq:no_CP_eg}, \ref{eq:1CP_eg}, \ref{eq:2CP_eg}).

The closer-the-better reward scheme of Eq.~\eqref{eq:R_ME_ctbr} is modified as
\begin{align}
    \label{eq:2CP_ctbr}
    R_{ij, kl} =
    \begin{cases}
        0,\ \hspace{90pt} i = N+1 \\
        \left( \frac{1}{\sqrt{2}} \right)^{|i-j|} \left( \frac{1}{\sqrt{2}} \right)^{|k-l|},\ i \ne N+1,\ j < k \\
        \left( \frac{1}{\sqrt{2}} \right)^{|i-k|} \left( \frac{1}{\sqrt{2}} \right)^{|j-l|},\ \text{otherwise}.
    \end{cases}
\end{align}

Figure~\ref{fig:beta'_beta''_pi_4_2CP} contrasts the solution of the reduced SDP (Eq.~\eqref{reduceddual}) when the variable $X$ is constrained to be Hermitian, and for the heuristic, we require the variable $X$ to satisfy Eq.~\eqref{eq:X_2CP}, as a function of the length of the sequence $N$.
Observe that for small sequence lengths up to length $13$, it may be disadvantageous to use the heuristic approach.
However for sequences of larger lengths, our heuristic is within $0.044$ of the solution, while being roughly \emph{eight} times faster to compute.

\begin{figure}[h]
    \centering
    \begin{subfigure}[b]{0.495\textwidth}
        \centering
        \includegraphics[width=\textwidth]{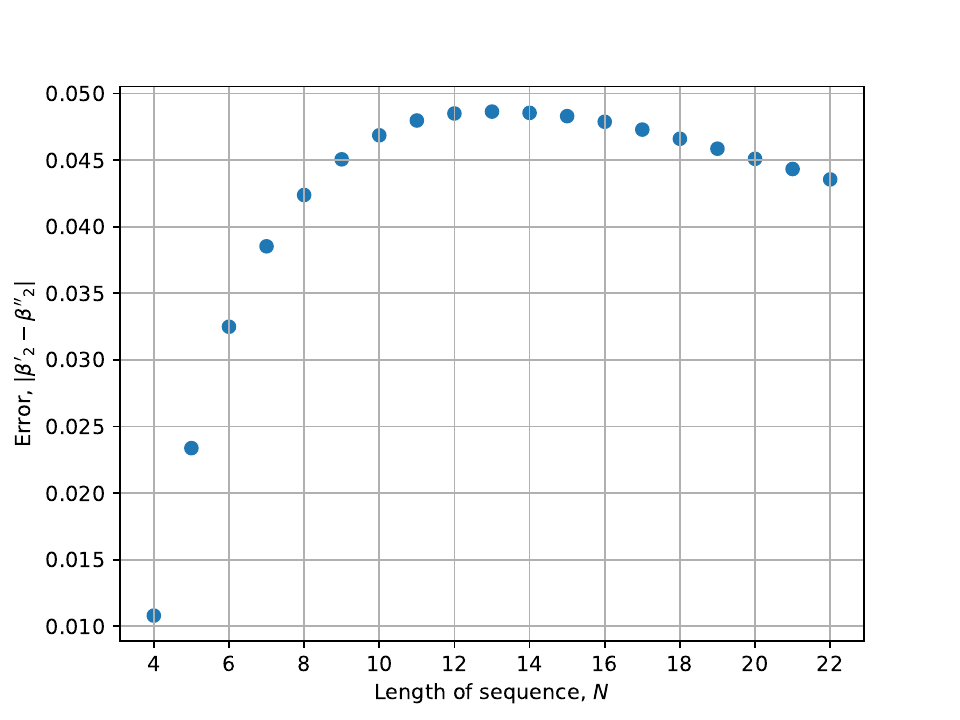}
        \caption{Absolute value of the difference between Eq.~\eqref{reduceddual} and its heuristic.}
    \end{subfigure}
    \hfill
    \begin{subfigure}[b]{0.495\textwidth}
        \centering
        \includegraphics[width=\textwidth]{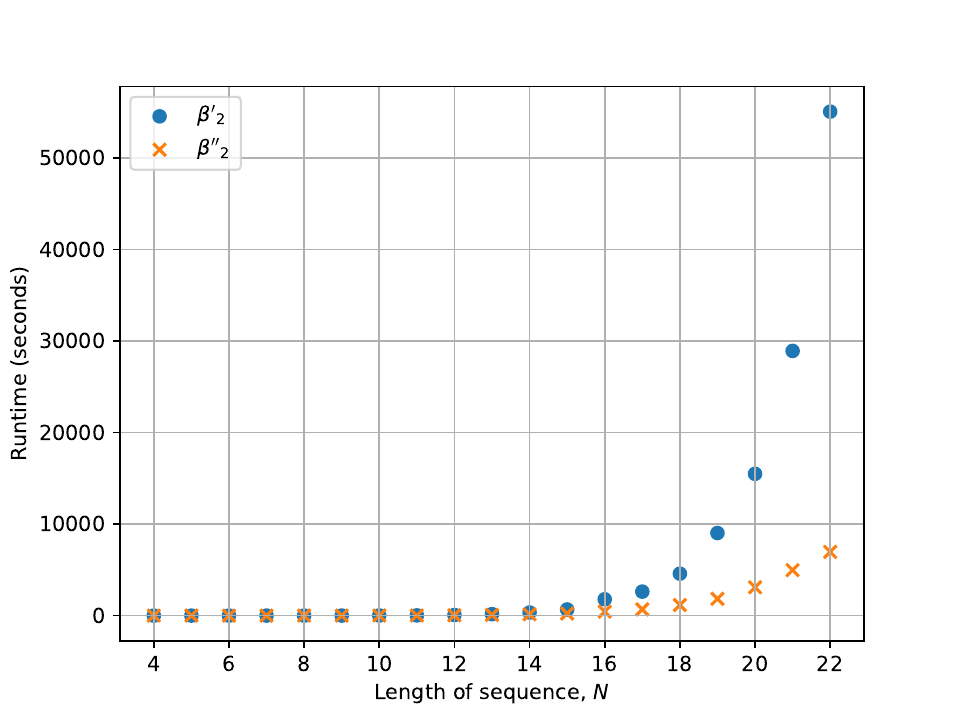}
        \caption{Runtimes of Eq.~\eqref{reduceddual} and its heuristic.}
    \end{subfigure}
    \caption{
    Here we look at the case where the state $\ket{0}$ mutated to the state $\ket{+}$ which in turn mutated to the state $\ket{1}$.
    For the reward described in Eq.~\eqref{eq:2CP_ctbr}, we depict the error between the heuristic and the solution (Eq.~\eqref{reduceddual}) (on the left), as well as the runtimes of both Eq.~\eqref{reduceddual} and the heuristic (on the right).
    For a short sequence up to length $13$, computing the heuristic might be a disadvantage.
    But for sequences of larger lengths, the heuristic is roughly \emph{eight} times faster than Eq.~\eqref{reduceddual} with a difference of only about $0.044$ between their computed values.
    To compute Eq.~\eqref{reduceddual} for sequence of length $22$, it takes about $15$ hours.
    If one were to continue this experiment for larger lengths of sequences, one should expect a behavior similar to Figure~\ref{fig:beta'_beta''_pi_4_1CP} (a).
    }
    \label{fig:beta'_beta''_pi_4_2CP}
\end{figure}

\paragraph{At most three changepoints.}
Finally, we consider all the possible sequential mutations, $\ket{0}$ to $\ket{+}$, $\ket{+}$ to $\ket{1}$, and $\ket{1}$ to $-\ket{-}$.
In other words, Bob must discriminate between all the possible states described in Eq.~(\ref{eq:no_CP_eg}, \ref{eq:1CP_eg}, \ref{eq:2CP_eg}, \ref{eq:3CP_eg}).

The reward strategy in accordance with the closer-the better rewards in Eq.~\eqref{eq:R_ME_ctbr} is
\begin{align}
    \label{eq:3CP_ctbr}
    R_{ijk, lmn} =
    \begin{cases}
        0,\ \hspace{140pt} i = N+1 \\
        \left( \frac{1}{\sqrt{2}} \right)^{|i-l|} \left( \frac{1}{\sqrt{2}} \right)^{|j-k|} \left( \frac{1}{\sqrt{2}} \right)^{|m-n|},\ i \ne N+1,\ l \le j \le k \le m \\
        \left( \frac{1}{\sqrt{2}} \right)^{|i-l|} \left( \frac{1}{\sqrt{2}} \right)^{|j-m|} \left( \frac{1}{\sqrt{2}} \right)^{|k-n|},\ i \ne N+1,\ l \le j \le n \text{ and } k > m \\
        \left( \frac{1}{\sqrt{2}} \right)^{|i-l|} \left( \frac{1}{\sqrt{2}} \right)^{|m-n|},\ \hspace{50pt} i \ne N+1,\ l \le m \le n < j \le k \\
        \left( \frac{1}{\sqrt{2}} \right)^{|i-j|} \left( \frac{1}{\sqrt{2}} \right)^{|l-k|} \left( \frac{1}{\sqrt{2}} \right)^{|m-n|},\ i \ne N+1,\ j < l \le k \le m \\
        \left( \frac{1}{\sqrt{2}} \right)^{|i-j|} \left( \frac{1}{\sqrt{2}} \right)^{|l-m|},\ \hspace{50pt} i \ne N+1,\ j < l \le k \text{ and } k > m \\
        \left( \frac{1}{\sqrt{2}} \right)^{|i-j|} \left( \frac{1}{\sqrt{2}} \right)^{|k-l|} \left( \frac{1}{\sqrt{2}} \right)^{|m-n|},\ i \ne N+1,\ j \le k < l
    \end{cases}. 
\end{align}

Figure~\ref{fig:beta'_beta''_pi_4_3CP} compares the runtime and the absolute difference in the solutions of the reduced SDP (Eq.~\eqref{reduceddual}) when the variable $X$ is constrained to be Hermitian, and the heuristic when the variable $X$ is constrained to satisfy Eq.~\eqref{eq:X_3CP}, as a function of the length of the sequence $N$.
Note here that for sequence of length $12$, solving Eq.~\eqref{reduceddual} involves computing a matrix of size $4096 \times 4096$.
Up to a sequence of length $8$, there is no gain in the runtime when computing the heuristic.
As $N$ gets larger, the heuristic is about \emph{ten} times faster than Eq.~\eqref{reduceddual} (on the right) with an error of roughly $0.083$ for $N = 12$.

\begin{figure}[htb]
    \centering
    \begin{subfigure}[b]{0.495\textwidth}
        \centering
        \includegraphics[width=\textwidth]{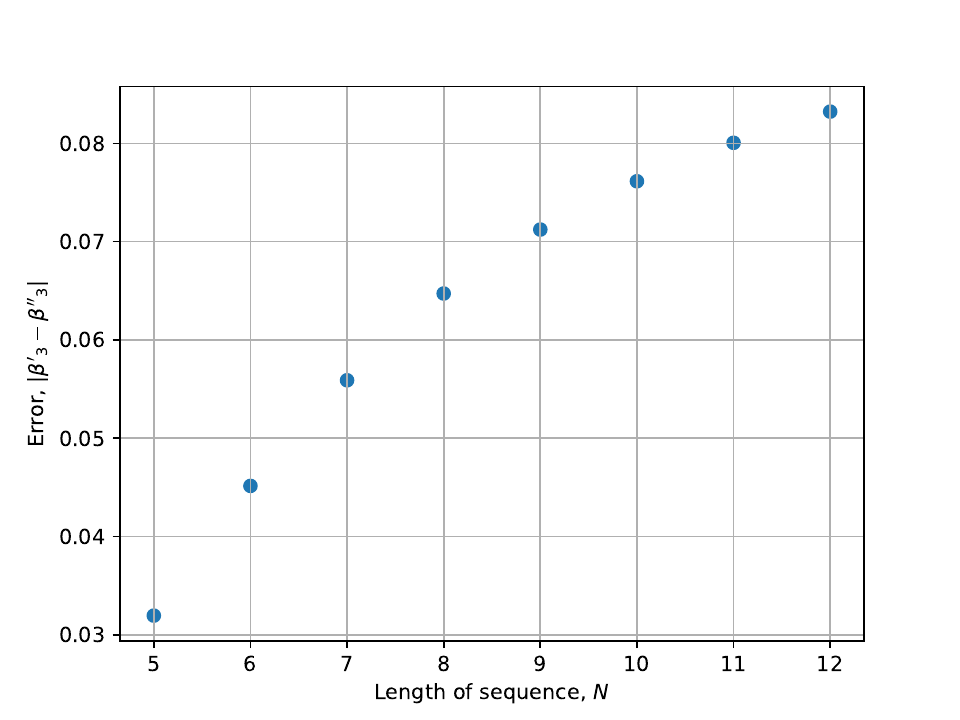}
        \caption{Absolute value of the difference between Eq.~\eqref{reduceddual} and its heuristic.}
    \end{subfigure}
    \hfill
    \begin{subfigure}[b]{0.495\textwidth}
        \centering
        \includegraphics[width=\textwidth]{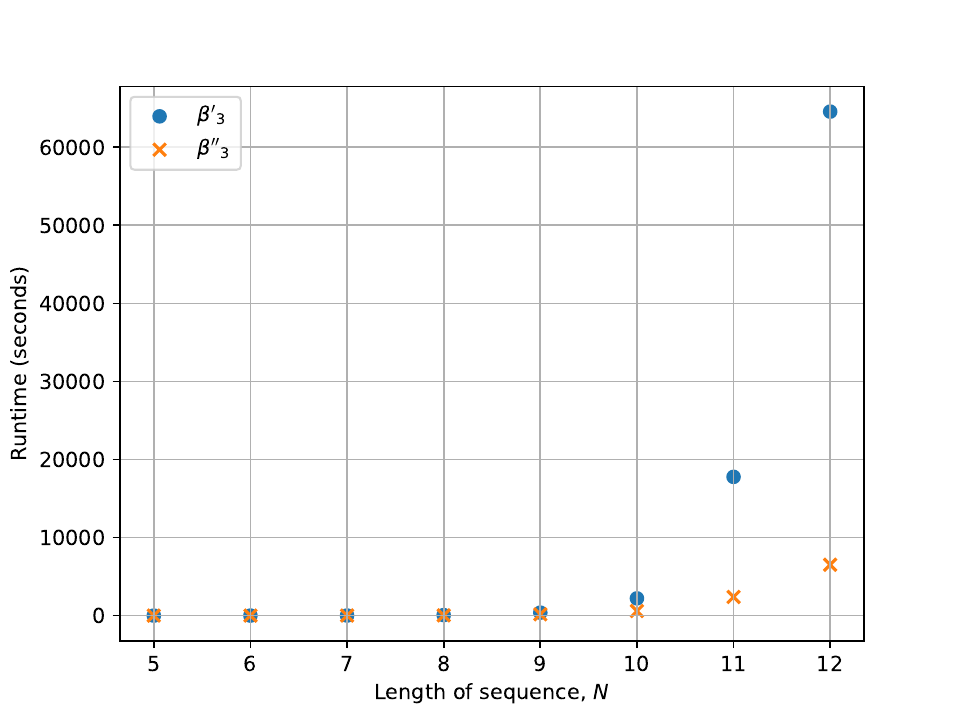}
        \caption{Runtimes of Eq.~\eqref{reduceddual} and its heuristic.}
    \end{subfigure}
    \caption{
    We look at sequences with three changepoints, each being sequences of $N$ states starting at $\ket{0}$ and first switching to $\ket{+}$, then to $\ket{1}$, and finally to $-\ket{-}$.
    Under the reward function of Eq.~\eqref{eq:3CP_ctbr}, we see the error (on the left) and the runtime (on the right) of both the solution Eq.~\eqref{reduceddual} and the heuristic.
    The SDP solving this heuristic is of size $\Ocomp(N^3)$.
    Comparing the runtimes we see that up to a sequence of $8$, there is no gain in the runtime by computing the heuristic.
    Beyond this, the heuristic is about \emph{ten} times faster than Eq.~\eqref{reduceddual} (on the right) with an error of roughly $0.083$ for $N = 12$.
    We terminate execution at $N = 12$ because computing Eq.~\eqref{reduceddual} takes $18$ hours.
    It would be reasonable to assume that for larger sequences, the error will resemble that of the one changepoint case illustrated in Figure~\ref{fig:beta'_beta''_pi_4_1CP} (a).
    }
    \label{fig:beta'_beta''_pi_4_3CP}
\end{figure}

We consider the previous discussion replacing these states $\ket{0}, \ket{+}, \ket{1}$ and $-\ket{-}$ with qubits differing by an angle of $\theta$ (Eq.~\eqref{eq:theta_states}) instead, in Appendix~\ref{sec:theta_expts}.

\section{Gram matrices and heuristics for multiple changepoints}

In this appendix we describe the scenarios for at most two and three changepoints. 

\subsection{At most two changepoints}\label{sec:2CP}

When dealing with the scenario where at most two changepoints can occur, Bob must focus on the following states
\begin{equation}
    \label{eq:2CP}
    \begin{aligned}
        \ket{\tau_{c_1 c_2}} &\coloneqq \ket{\tau_{c_1 c_2 N \cdots N}} = \ket{\psi}^{\otimes c_1}\ \otimes \ket{\phi_1}^{\otimes (c_2 - c_1)}\ \otimes \ket{\phi_2}^{\otimes (N - c_2)}, \quad (c_1, c_2) \in \Iset_N.
    \end{aligned}
\end{equation}

Note here that
\begin{align}
    \ket{\tau_{c_1}} \coloneqq \ket{\tau_{c_1 N}} = \ket{\psi}^{\otimes c_1}\ \otimes \ket{\phi_1}^{\otimes (N - c_1)}
\end{align}
denotes the sequences corresponding to a single changepoint and
\begin{align}
    \ket{\tau_N} \coloneqq \ket{\tau_{N N}} = \ket{\psi}^{\otimes N} 
\end{align}
indicates the sequence with no changepoint.

\paragraph{Gram matrix.}

The Gram matrix can be computed using the formula,

\begin{align}
    \label{eq:T_ijkl}
    (T_{2CP})_{ij,kl} =
    \begin{cases}
        \gamma_1^{|i-j|}\ \gamma_2^{|j-k|}\ \gamma_{12}^{|k-l|}, &j < k \\
        \gamma_1^{|i-k|}\ \gamma_{12}^{|j-l|},\ \hspace{30pt} &\text{otherwise}
    \end{cases} 
\end{align}
where we have assumed for simplicity that the overlaps, $\gamma_1, \gamma_2$ and $\gamma_{12}$ are all non-negative real numbers. 
We proved that this assumption is without loss of generality in the one changepoint case (Lemma~\ref{lem:T_sym_Toep}), assuming it here makes the technical analysis a bit tidier. 

When represented as a matrix, we obtain a symmetric block matrix as illustrated in Eq.~(\ref{eq:T_2CP}).
\begin{align}
    \label{eq:T_2CP}
    T_{2CP} =
    \begin{pmatrix}
        T_{11} & T_{12} & T_{13} & \cdots & T_{1,N-2} & T_{1,N-1} & \rvline & T_{1N} \\
        T_{12}^\top & T_{11}^{(1)} & T_{12}^{(1)} & \cdots & T_{1,N-3}^{(1)} & T_{1,N-2}^{(1)} & \rvline & T_{2N} \\
        T_{13}^\top & T_{12}^{(1)\top} & T_{11}^{(2)} & \cdots & T_{1,N-4}^{(2)} & T_{1, N-3}^{(2)} & \rvline & T_{3N} \\
        \vdots & \vdots & \vdots & \ddots & \vdots & \vdots & \rvline & \vdots \\
        T_{1,N-2}^\top & T_{1,N-3}^{(1)\top} & T_{1,N-4}^{(2)\top} & \cdots & T_{11}^{(N-2)} & T_{12}^{(N-2)} & \rvline & T_{N-2,N} \\
        T_{1,N-1}^\top & T_{1,N-2}^{(1)\top} & T_{1,N-3}^{(2)\top} & \cdots & T_{12}^{(N-2)\top} & T_{11}^{(N-1)} & \rvline & T_{N-1,N} \\
        \hline
        T_{1N}^\top & T_{2N}^\top & T_{3N}^\top & \cdots & T_{N-2,N}^\top & T_{N-1,N}^\top & \rvline & 1 \\
    \end{pmatrix}. 
\end{align}
Each block, denoted by $T_{ik}$, is an $(N-i) \times (N-k)$ Toeplitz matrix.
The block $T_{11}$ is a symmetric Toeplitz matrix where each entry $(T_{11})_{jl} = \gamma_{12}^{|j-l|}$.
The notation $T_{ik}^{(n)}$ indicates the matrix obtained by discarding the \emph{last} $n$ rows and the \emph{last} $n$ columns from $T_{ik}$.

\paragraph{Heuristic solution.}

Exploiting this structure, we consider a heuristic approach by restricting the solution $X'_{2CP}$ to have a similar structure
\begin{align}
    \label{eq:X_2CP}
    X'_{2CP} =
    \begin{pmatrix}
        X_{11} & X_{12} & X_{13} & \cdots & X_{1,N-2} & X_{1,N-1} & \rvline & \\
        X_{12}^\top & X_{11}^{(1)} & X_{12}^{(1)} & \cdots & X_{1,N-3}^{(1)} & X_{1,N-2}^{(1)} & \rvline & \\
        X_{13}^\top & X_{12}^{(1)\top} & X_{11}^{(2)} & \cdots & X_{1,N-4}^{(2)} & X_{1,N-3}^{(2)} & \rvline \\
        \vdots & \vdots & \vdots & \ddots & \vdots & \vdots & \rvline & x_N \\
        X_{1,N-2}^\top & X_{1,N-3}^{(1)\top} & X_{1,N-4}^{(2)\top} & \cdots & X_{11}^{(N-2)} & X_{12}^{(N-2)} & \rvline & \\
        X_{1,N-1}^\top & X_{1,N-2}^{(1)\top} & X_{1,N-3}^{(2)\top} & \cdots & X_{12}^{(N-2)\top} & X_{11}^{(N-1)} & \rvline & \\
        \hline
         &  &  & x_N^\top &  &  & \rvline & z \\ 
    \end{pmatrix},
\end{align}
where $X_{ik}$ is an $(N-i) \times (N-k)$ Toeplitz matrix.
The matrix $X_{11}$ is a symmetric Toeplitz matrix.
The remaining blocks $X_{ik}^{(n)}$ follow the same convention as its counterpart $T_{ik}^{(n)}$. 
Here $z$ denotes the last entry of the vector $x_N$.

\subsection{At most three changepoints} \label{sec:3CP}

If we have at most three changepoints, Bob has to consider the following states
\begin{equation}
    \label{eq:3CP}
    \begin{aligned}
        \ket{\tau_{c_1 c_2 c_3}} &\coloneqq \ket{\tau_{c_1 c_2 c_3 N \cdots N}} = \ket{\psi}^{\otimes c_1}\ \otimes \ket{\phi_1}^{\otimes (c_2 - c_1)}\ \otimes \ket{\phi_2}^{\otimes (c_3 - c_2)}\ \otimes \ket{\phi_3}^{\otimes (N - c_3)}, \quad (c_1, c_2, c_3) \in \Iset_N.
    \end{aligned}
\end{equation}
Note here that
\begin{align}
    \ket{\tau_{c_1 c_2}} \coloneqq \ket{\tau_{c_1 c_2 N}} = \ket{\psi}^{\otimes c_1}\ \otimes \ket{\phi_1}^{\otimes (c_2 - c_1)}\ \otimes \ket{\phi_2}^{\otimes (N - c_2)}
\end{align}
is the sequence for at most two changepoints,
\begin{align}
    \ket{\tau_{c_1}} \coloneqq \ket{\tau_{c_1 N N}} = \ket{\psi}^{\otimes c_1}\ \otimes \ket{\phi_1}^{\otimes (N - c_1)}
\end{align}
denotes the sequences corresponding to a single changepoint and
\begin{align}
    \ket{\tau_N} \coloneqq \ket{\tau_{N N N}} = \ket{\psi}^{\otimes N} 
\end{align}
indicates the sequence with no changepoint.

\paragraph{Gram matrix.}

\begin{align}
    T_{ijk,lmn} =
    \begin{cases}
        \gamma_1^{|i-l|}\ \gamma_{12}^{|j-k|}\ \gamma_{13}^{|k-m|} \gamma_{23}^{|m-n|},\ \hspace{30pt} l \le j \le k \le m \\
        \gamma_1^{|i-l|}\ \gamma_{12}^{|j-m|}\ \gamma_{23}^{|k-n|},\ \hspace{60pt} l \le j \le n \text{ and } k > m \\
        \gamma_1^{|i-l|}\ \gamma_{12}^{|m-n|}\ \gamma_{13}^{|n-j|},\ \hspace{60pt} l \le m \le n < j \le k \\
        \gamma_1^{|i-j|}\ \gamma_2^{|j-l|}\ \gamma_{12}^{|l-k|}\ \gamma_{13}^{|k-m|}\ \gamma_{23}^{|m-n|},\ j < l \le k \le m \\
        \gamma_1^{|i-j|}\ \gamma_2^{|j-l|}\ \gamma_{12}^{|l-m|},\ \hspace{62pt} j < l \le k \text{ and } k > m \\
        \gamma_1^{|i-j|}\ \gamma_2^{|j-k|}\ \gamma_3^{|k-l|}\ \gamma_{13}^{|l-m|}\ \gamma_{23}^{|m-n|},\ j \le k < l
    \end{cases}. 
\end{align}

In this formula, we assume that $i \le l$.
It is straightforward to compute the other entries since the block matrix is symmetric.

The Gram matrix can be depicted as the block matrix denoted in Eq.~(\ref{eq:T_3CP}).
\begin{align}
    \label{eq:T_3CP}
    T_{3CP} =
    \begin{pmatrix}
        V_{11} & V_{12} & V_{13} & \cdots & V_{1,N-2} & V_{1,N-1} & \rvline & V_{1N} \\
        V_{12}^\top & V_{11}^{[1]} & V_{12}^{[1]} & \cdots & V_{1,N-3}^{[1]} & V_{1,N-2}^{[1]} & \rvline & V_{2N} \\
        V_{13}^\top & V_{12}^{[1]\top} & V_{11}^{[2]} & \cdots & V_{1,N-4}^{[2]} & V_{1, N-3}^{[2]} & \rvline & V_{3N} \\
        \vdots & \vdots & \vdots & \ddots & \vdots & \vdots & \rvline & \vdots \\
        V_{1,N-2}^\top & V_{1,N-3}^{[1]\top} & V_{1,N-4}^{[2]\top} & \cdots & V_{11}^{[N-3]} & V_{12}^{[N-3]} & \rvline & V_{N-2,N} \\
        V_{1,N-1}^\top & V_{1,N-2}^{[1]\top} & V_{1,N-3}^{[2]\top} & \cdots & V_{12}^{[N-3]\top} & V_{11}^{[N-2]} & \rvline & V_{N-1,N} \\
        \hline
        V_{1N}^\top & V_{2N}^\top & V_{3N}^\top & \cdots & V_{N-2,N}^\top & V_{N-1,N}^\top & \rvline & 1 \\
    \end{pmatrix}. 
\end{align}
Each block $V_{ik}$ is a matrix containing $(N-i) \times (N-k)$ blocks.
The notation $V_{ik}^{[m]}$ indicates the block matrix obtained by discarding the blocks along the \emph{first} $m$ rows and the \emph{first} $m$ columns of $V_{ik}$.

Determining the matrices along the \emph{first} row and the \emph{last} column of $T_{3CP}$ is sufficient to construct the entire matrix.
Eqs.~\eqref{eq:V_11}-\eqref{eq:V_2N} describe how to construct each of the matrices $V_{1k}$ and $V_{iN}$ for $i, k \in \{1, \dots, N-1\}$.
Each block in this block matrix $S_{ij,kl}$ is constructed in the same manner as was used by $T_{ik}$ in Section~\ref{sec:2CP}.
The matrix $S_{12,12}$ is symmetric Toeplitz where each entry is given as $(S_{12,12})_{jl} = \gamma_{23}^{|j-l|}$.
We therefore have the following matrices

\begin{align}
    \label{eq:V_11}
    V_{11} =
    \begin{pmatrix}
        S_{12,12} & S_{12,13} & S_{12,14} & \cdots & S_{12,1,N-2} & S_{12,1,N-1} & \rvline & S_{12,1N} \\
        S_{12,13}^\top & S_{12,12}^{(1)} & S_{12,13}^{(1)} & \cdots & S_{12,1,N-3}^{(1)} & S_{12,1,N-2}^{(1)} & \rvline & S_{13,1N} \\
        S_{12,14}^\top & S_{12,13}^{(1)\top} & S_{12,12}^{(2)} & \cdots & S_{12,1,N-4}^{(2)} & S_{12,1,N-3}^{(2)} & \rvline & S_{14,1N} \\
        \vdots & \vdots & \vdots & \ddots & \vdots & \vdots & \rvline & \vdots \\
        S_{12,1,N-2}^\top & S_{12,1,N-3}^{(1)\top} & S_{12,1,N-4}^{(2)\top} & \cdots & S_{12,12}^{(N-4)} & S_{12,13}^{(N-4)} & \rvline & S_{1,N-2,1N} \\
        S_{12,1,N-1}^\top & S_{12,1,N-2}^{(1)\top} & S_{12,1,N-3}^{(2)\top} & \cdots & S_{12,13}^{(N-4)\top} & S_{12,12}^{(N-3)} & \rvline & S_{1,N-1,1N} \\
        \hline
        S_{12,1N}^\top & S_{13,1N}^\top & S_{14,1N}^\top & \cdots & S_{1,N-2,1N}^\top & S_{1,N-1,1N} & \rvline & 1 \\
    \end{pmatrix}
\end{align}

\begin{align}
    \label{eq:V_12}
    V_{12} =
    \begin{pmatrix}
        S_{12,23} & S_{12,24} & S_{12,25} & \cdots & S_{12,2,N-2} & S_{12,2,N-1} & \rvline & S_{12,2N} \\
        S_{13,23} & S_{12,23}^{(1)} & S_{12,24}^{(1)} & \cdots & S_{12,2,N-3}^{(1)} & S_{12,2,N-2}^{(1)} & \rvline & S_{13,2N} \\
        S_{12,23}^{(1)\top} & S_{13,23}^{(1)} & S_{12,23}^{(2)} & \cdots & S_{12,2,N-4}^{(2)} & S_{12,2,N-3}^{(2)} & \rvline & S_{14,2N} \\
        \vdots & \vdots & \vdots & \ddots & \vdots & \vdots & \rvline & \vdots \\
        S_{12,2,N-3}^{(1)\top} & S_{12,2,N-4}^{(2)\top} & S_{12,2,N-5}^{(2)\top} & \cdots & S_{13,23}^{(N-5)} & S_{12,23}^{(N-5)} & \rvline & S_{1,N-2,2N} \\
        S_{12,2,N-2}^{(1)\top} & S_{12,2,N-3}^{(2)\top} & S_{12,2,N-4}^{(2)\top} & \cdots & S_{12,23}^{(N-5)\top} & S_{13,23}^{(N-4)} & \rvline & S_{1,N-1,2N} \\
        \hline
        S_{13,2N}^\top & S_{14,2N}^\top & S_{15,2N}^\top & \cdots & S_{1,N-2,2N}^\top & S_{1,N-1,2N}^\top & \rvline & S_{1N,2N}
    \end{pmatrix}
\end{align}

\begin{align}
    \label{eq:V_13}
    V_{13} = 
    \begin{pmatrix}
        S_{12,34} & S_{12,35} & S_{12,36} & \cdots & S_{12,3,N-2} & S_{12,3,N-1} & \rvline & S_{12,3N} \\
        S_{13,34} & S_{12,34}^{(1)} & S_{12,35}^{(1)} & \cdots & S_{12,3,N-3}^{(1)} & S_{12,3,N-2}^{(1)} & \rvline & S_{13,3N} \\
        S_{14,34} & S_{13,34}^{(1)} & S_{12,34}^{(1)} & \cdots & S_{12,3,N-4}^{(2)} & S_{12,3,N-3}^{(2)} & \rvline & S_{14,3N} \\
        \vdots & \vdots & \vdots & \ddots & \vdots & \vdots & \rvline & \vdots \\
        S_{12,3,N-4}^{(2)\top} & S_{12,3,N-5}^{(3)\top} & S_{12,3,N-6}^{(4)\top} & \cdots & S_{14,34}^{(N-6)} & S_{13,34}^{(N-6)} & \rvline & S_{1,N-2,3N} \\
        S_{12,3,N-3}^{(2)\top} & S_{12,3,N-4}^{(3)\top} & S_{12,3,N-5}^{(4)\top} & \cdots & S_{13,34}^{(N-6)\top} & S_{14,34}^{(N-5)} & \rvline & S_{1,N-1,3N} \\
        \hline
        S_{14,3N}^\top & S_{15,3N}^\top & S_{16,3N}^\top & \cdots & S_{1,N-2,3N}^\top & S_{1,N-1,3N}^\top & \rvline & S_{1N,3N}
    \end{pmatrix}
\end{align}

\begin{align}
    \label{eq:V_1N_2}
    V_{1,N-2} =
    \begin{pmatrix}
        S_{12,N-2,N-1} & \rvline & S_{12,N-2,N} \\
        S_{13,N-2,N-1} & \rvline & S_{13,N-2,N} \\
        S_{14,N-2,N-1} & \rvline & S_{14,N-2,N} \\
        \vdots & \rvline & \vdots \\
        S_{1,N-2,N-2,N-1} & \rvline & S_{1,N-2,N-2,N} \\
        S_{1,N-1,N-2,N-1} & \rvline & S_{1,N-1,N-2,N} \\
        S_{1N,N-2,N-1} & \rvline & S_{1N,N-2,N} \\
    \end{pmatrix}
    \quad
    V_{1,N-1} = 
    \begin{pmatrix}
        S_{12,N-1,N} \\
        S_{13,N-1,N} \\
        S_{14,N-1,N} \\
        \vdots \\
        S_{1,N-2,N-1,N} \\
        S_{1,N-1,N-1,N} \\
        S_{1N,N-1,N} \\
    \end{pmatrix}
    \quad
    V_{1N} =
    \begin{pmatrix}
        S_{12,NN} \\
        S_{13,NN} \\
        S_{14,NN} \\
        \vdots \\
        S_{1,N-2,NN} \\
        S_{1,N-1,NN} \\
        S_{1N,NN} \\
    \end{pmatrix}
\end{align}

\begin{align}
    \label{eq:V_2N}
    V_{2N} =
    \begin{pmatrix}
        S_{23,NN} \\
        S_{24,NN} \\
        \vdots \\
        S_{2,N-2,NN} \\
        S_{2,N-1,NN} \\
        S_{2N,NN} \\
    \end{pmatrix},
    \
    V_{3N} =
    \begin{pmatrix}
        S_{34,NN} \\
        \vdots \\
        S_{3,N-2,NN} \\
        S_{3,N-1,NN} \\
        S_{3N,NN} \\
    \end{pmatrix},
    \
    V_{N-2,N} =
    \begin{pmatrix}
        S_{N-2,N-1,NN} \\
        S_{N-2,N,NN}
    \end{pmatrix},
    \
    V_{N-1,N} =
    \begin{pmatrix}
        S_{N-1,N,NN}
    \end{pmatrix}.
\end{align}

\paragraph{Heuristic approach.}

We once again exploit the Toeplitz-like structure of the Gram matrix and restrict the solution $X_{3CP}$ to have a similar structure as below
\begin{align}
    \label{eq:X_3CP}
    X'_{3CP} =
    \begin{pmatrix}
        Y_{11} & Y_{12} & Y_{13} & \cdots & Y_{1,N-2} & Y_{1,N-1} & \rvline & \\
        Y_{12}^\top & Y_{11}^{[1]} & Y_{12}^{[1]} & \cdots & Y_{1,N-3}^{[1]} & Y_{1,N-2}^{[1]} & \rvline & \\
        Y_{13}^\top & Y_{12}^{[1]\top} & Y_{11}^{[2]} & \cdots & Y_{1,N-4}^{[2]} & Y_{1,N-3}^{[2]} & \rvline \\
        \vdots & \vdots & \vdots & \ddots & \vdots & \vdots & \rvline & y_N \\
        Y_{1,N-2}^\top & Y_{1,N-3}^{[1]\top} & Y_{1,N-4}^{[2]\top} & \cdots & Y_{11}^{[N-3]} & Y_{12}^{[N-3]} & \rvline & \\
        Y_{1,N-1}^\top & Y_{1,N-2}^{[1]\top} & Y_{1,N-3}^{[2]\top} & \cdots & Y_{12}^{[N-3]\top} & Y_{11}^{[N-2]} & \rvline & \\
        \hline
         &  &  & y_N^\top &  &  & \rvline & z \\
    \end{pmatrix}
\end{align}

\begin{align}
    \label{eq:Y_11}
    Y_{11} =
    \begin{pmatrix}
        W_{12,12} & W_{12,13} & W_{12,14} & \cdots & W_{12,1,N-2} & W_{12,1,N-1} & \rvline & \\
        W_{12,13}^\top & W_{12,12}^{(1)} & W_{12,13}^{(1)} & \cdots & W_{12,1,N-3}^{(1)} & W_{12,1,N-2}^{(1)} & \rvline & \\
        W_{12,14}^\top & W_{12,13}^{(1)\top} & W_{12,12}^{(2)} & \cdots & W_{12,1,N-4}^{(2)} & W_{12,1,N-3}^{(2)} & \rvline & \\
        \vdots & \vdots & \vdots & \ddots & \vdots & \vdots & \rvline & w_{1N} \\
        W_{12,1,N-2}^\top & W_{12,1,N-3}^{(1)\top} & W_{12,1,N-4}^{(2)\top} & \cdots & W_{12,12}^{(N-4)} & W_{12,13}^{(N-4)} & \rvline & \\
        W_{12,1,N-1}^\top & W_{12,1,N-2}^{(1)\top} & W_{12,1,N-3}^{(2)\top} & \cdots & W_{12,13}^{(N-4)\top} & W_{12,12}^{(N-3)} & \rvline & \\
        \hline
         &  &  & w_{1N}^\top &  &  & \rvline & z_{1N} \\
    \end{pmatrix}
\end{align}

\begin{align}
    \label{eq:Y_12}
    Y_{12} =
    \begin{pmatrix}
        W_{12,23} & W_{12,24} & W_{12,25} & \cdots & W_{12,2,N-2} & W_{12,2,N-1} & \rvline & \\
        W_{13,23} & W_{12,23}^{(1)} & W_{12,24}^{(1)} & \cdots & W_{12,2,N-3}^{(1)} & W_{12,2,N-2}^{(1)} & \rvline &  \\
        W_{12,23}^{(1)\top} & W_{13,23}^{(1)} & W_{12,23}^{(2)} & \cdots & W_{12,2,N-4}^{(2)} & W_{12,2,N-3}^{(2)} & \rvline &  \\
        \vdots & \vdots & \vdots & \ddots & \vdots & \vdots & \rvline & w_{2N} \\
        W_{12,2,N-3}^{(1)\top} & W_{12,2,N-4}^{(2)\top} & W_{12,2,N-5}^{(2)\top} & \cdots & W_{13,23}^{(N-5)} & W_{12,23}^{(N-4)} & \rvline &  \\
        W_{12,2,N-2}^{(1)\top} & W_{12,2,N-3}^{(2)\top} & W_{12,2,N-4}^{(2)\top} & \cdots & W_{12,23}^{(N-4)\top} & W_{13,23}^{(N-4)} & \rvline &  \\
        \hline
         &  &  & w_{2N}^\top &  &  & \rvline & z_{2N} \\
    \end{pmatrix}
\end{align}

\begin{align}
    \label{eq:Y_13}
    Y_{13} =
    \begin{pmatrix}
        W_{12,34} & W_{12,35} & W_{12,36} & \cdots & W_{12,3,N-2} & W_{12,3,N-1} & \rvline &  \\
        W_{13,34} & W_{12,34}^{(1)} & W_{12,35}^{(1)} & \cdots & W_{12,3,N-3}^{(1)} & W_{12,3,N-2}^{(1)} & \rvline &  \\
        W_{14,34} & W_{13,34}^{(1)} & W_{12,34}^{(1)} & \cdots & W_{12,3,N-4}^{(2)} & W_{12,3,N-3}^{(2)} & \rvline &  \\
        \vdots & \vdots & \vdots & \ddots & \vdots & \vdots & \rvline & w_{3N} \\
        W_{12,3,N-4}^{(2)\top} & W_{12,3,N-5}^{(3)\top} & W_{12,3,N-6}^{(4)\top} & \cdots & W_{14,34}^{(N-6)} & W_{13,34}^{(N-6)} & \rvline &  \\
        W_{12,3,N-3}^{(2)\top} & W_{12,3,N-4}^{(3)\top} & W_{12,3,N-5}^{(4)\top} & \cdots & W_{13,34}^{(N-6)\top} & W_{14,34}^{(N-5)} & \rvline &  \\
        \hline
         &  &  & w_{3N}^\top &  &  & \rvline & z_{3N}
    \end{pmatrix}
\end{align}

\begin{align}
    \label{eq:Y_1N_2}
    Y_{1,N-2} =
    \begin{pmatrix}
        W_{12,N-2,N-1} & \rvline &  \\
        W_{13,N-2,N-1} & \rvline &  \\
        W_{14,N-2,N-1} & \rvline &  \\
        \vdots & \rvline & w_{N-2,N} \\
        W_{1,N-2,N-2,N-1} & \rvline &  \\
        W_{1,N-1,N-2,N-1} & \rvline &  \\
        W_{1N,N-2,N-1} & \rvline &  \\
    \end{pmatrix},
\end{align}

where $Y_{1,N-1}$ is a vector of size $N-1$, while $y_N$ is a vector of size $N + \frac{N (N-1) (N-2)}{6}$.

\subsection{At most $P$ changepoints}

Observe that $T_{11}$ in Eq.~\eqref{eq:T_2CP} is constructed in the exact same manner as $T_{1CP}$ used in the one changepoint case.
Similarly, $V_{11}$ in Eq.~\eqref{eq:T_3CP} follows the same construction as $T_{2CP}$ described in Eq.~\eqref{eq:T_2CP}.
Therefore, we guess that a heuristic solution $X'$ to solve the general case can be constructed by exploiting the structure of the corresponding Gram matrix in a similar manner as the cases we examined.

\section{Numerical experiments for changepoints with varying overlaps}\label{sec:theta_expts}

We reproduce some of the numerical experiments in section~\ref{sec:expts} with states differing by an angle $\theta$ (section~\ref{sec:expts} considered $\theta = \pi/4$).
Note that the rewards also change from $1/\sqrt{2}$ to $\cos(\theta)$ in Eqs.~(\ref{eq:1CP_ctbr}, \ref{eq:2CP_ctbr}, \ref{eq:3CP_ctbr}).

Figures~\ref{fig:1CP},~\ref{fig:2CP} and~\ref{fig:3CP} illustrate the optimal reward function and our heuristic for one, two, and three changepoints when $\theta$ takes values in $\{\pi/8,\ \pi/4,\ 3\pi/8\}$.

One thing we note is that the SDP solver is unable to find the solution for the three changepoint scenario when $\theta = \pi/8$ and hence returns an arbitrary value of $-1$ for sequences of length greater than $5$.
Note that even when this occurs, our heuristic approach is able to compute a solution.

A general observation from these figures is that for states that are farther apart, the absolute difference between the heuristic and the solution decreases very rapidly.
As the states get closer to one another, this decrease in the difference happens at a much slower rate.
A possible explanation for this is that when the states are farther apart to begin with, it is not hard to distinguish them as is.
As the length of the sequence increases, the ease of distinguishability also increases.
As the states get closer to one another, the sequence must be fairly large before it becomes easy to distinguish them.

\begin{figure}[h]
    \centering
    \begin{subfigure}[b]{0.495\textwidth}
        \centering
        \includegraphics[width=\textwidth]{prob_1CP.pdf}
        \caption{Optimal reward function value as a function of $N$ for $\theta \in \{\pi/8, \pi/4, 3\pi/8\}$.}
    \end{subfigure}
    \hfill
    \begin{subfigure}[b]{0.495\textwidth}
        \centering
        \includegraphics[width=\textwidth]{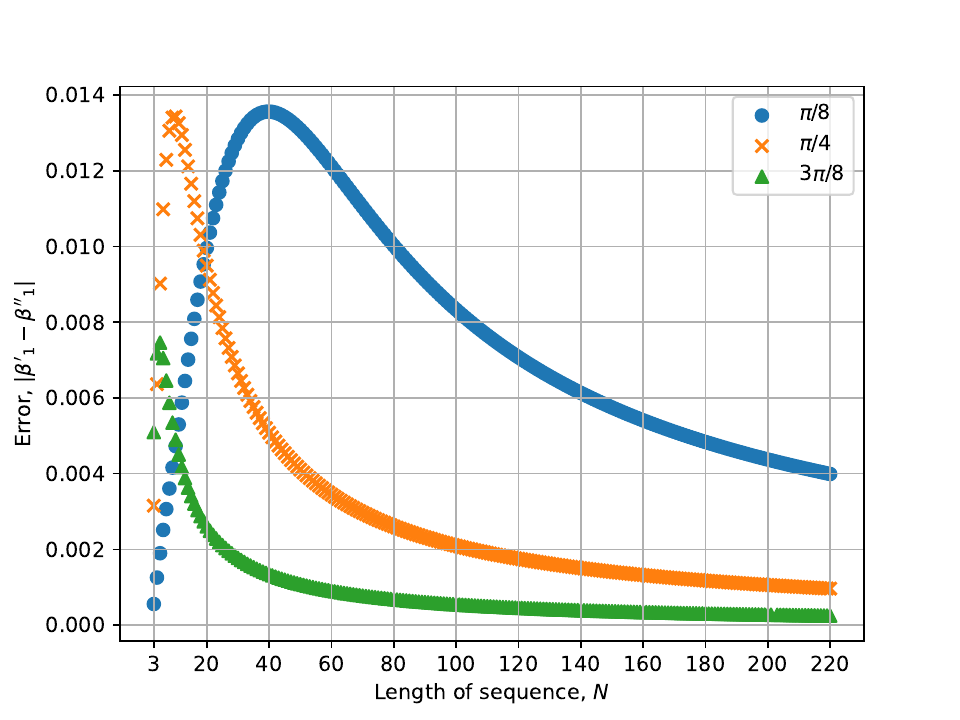}
        \caption{Absolute value of the difference between the solution and the heuristic for $\theta \in \{\pi/8, \pi/4, 3\pi/8\}$ as a function of $N$.}
    \end{subfigure}
    \caption{
    Consider the sequences when Alice promises Bob $N$ copies of the state $\ket{0}$ which mutate to the state $\ket{\phi_1} = \cos(\theta) + \sin(\theta)$ for $\theta \in \{\pi/8, \pi/4, 3\pi/8\}$.
    We compare the optimal reward function values for varying $\theta$ as a function of the length of the sequence $N$ (on the left), and the absolute difference between the heuristic and the solution (on the right).
    Observe that for states that are farther apart ($\theta = 3\pi/8$), the error approaches zero rather quickly.
    As the states get closer, it takes much longer for the difference to decrease.
    }
    \label{fig:1CP}
\end{figure}

\begin{figure}[h]
    \centering
    \begin{subfigure}[b]{0.495\textwidth}
        \centering
        \includegraphics[width=\textwidth]{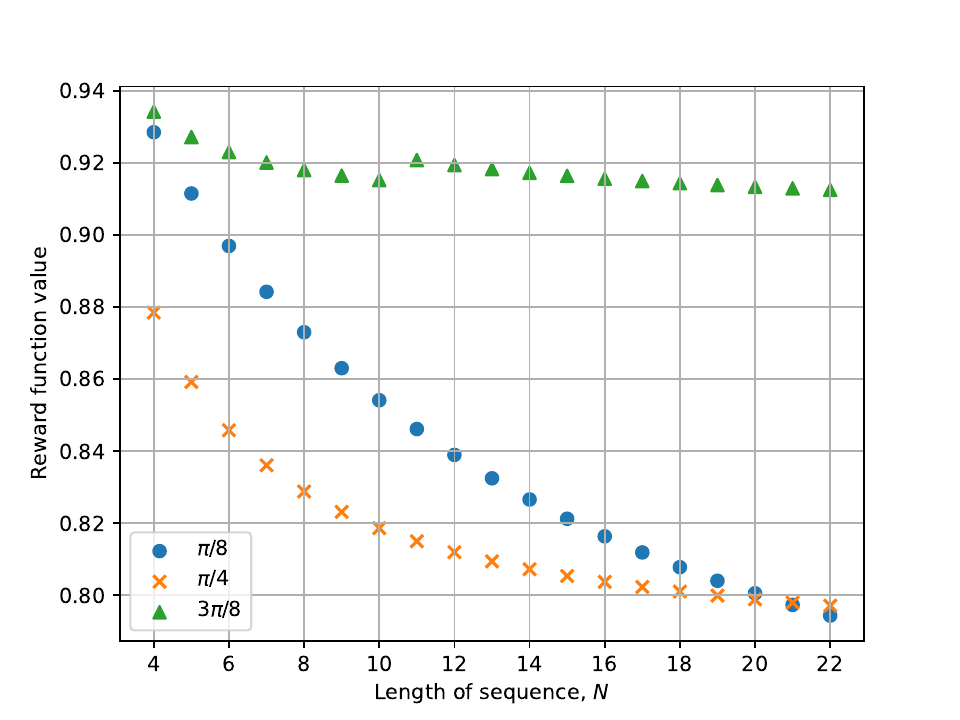}
        \caption{Optimal reward function value as a function of $N$ for $\theta \in \{\pi/8, \pi/4, 3\pi/8\}$.}
    \end{subfigure}
    \hfill
    \begin{subfigure}[b]{0.495\textwidth}
        \centering
        \includegraphics[width=\textwidth]{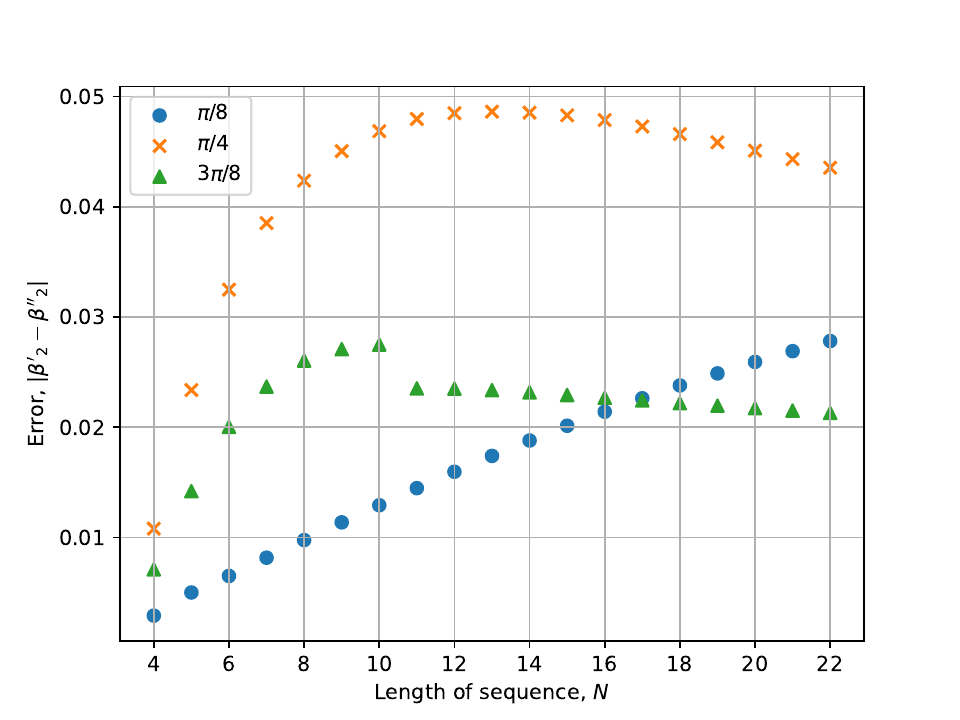}
        \caption{Absolute value of the difference between the solution and the heuristic for $\theta \in \{\pi/8, \pi/4, 3\pi/8\}$ as a function of $N$.}
    \end{subfigure}
    \caption{
    Here we consider the scenario where the state $\ket{0}$ first mutates to the state \\ $\ket{\phi_1} = \cos(\theta) + \sin(\theta)$, which further mutates to the state $\ket{\phi_2} = \cos(2\theta) + \sin(2\theta)$, for \\ $\theta \in \{\pi/8, \pi/4, 3\pi/8\}$.
    The optimal reward function values for varying $\theta$ as a function of the length of the sequence $N$ is first compared (on the left), followed by a comparison of the absolute difference between the heuristic and the solution (on the right).
    The behavior is similar to the one changepoint case.
    }
    \label{fig:2CP}
\end{figure}

\begin{figure}[h]
    \centering
    \begin{subfigure}[b]{0.495\textwidth}
        \centering
        \includegraphics[width=\textwidth]{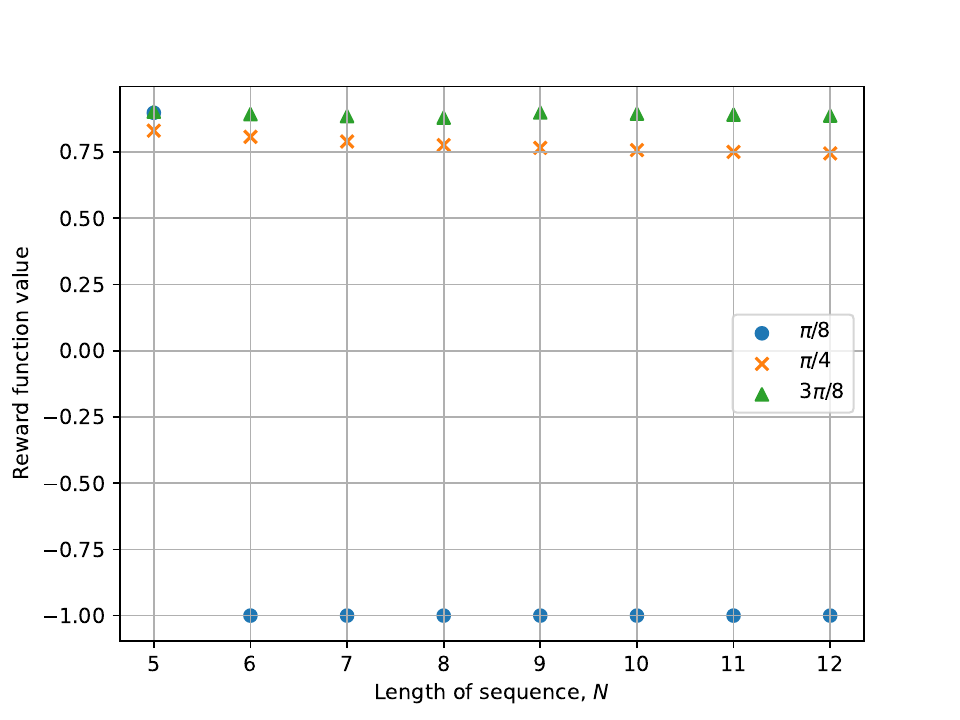}
        \caption{Optimal reward function value as a function of $N$ for $\theta \in \{\pi/8, \pi/4, 3\pi/8\}$.}
    \end{subfigure}
    \hfill
    \begin{subfigure}[b]{0.495\textwidth}
        \centering
        \includegraphics[width=\textwidth]{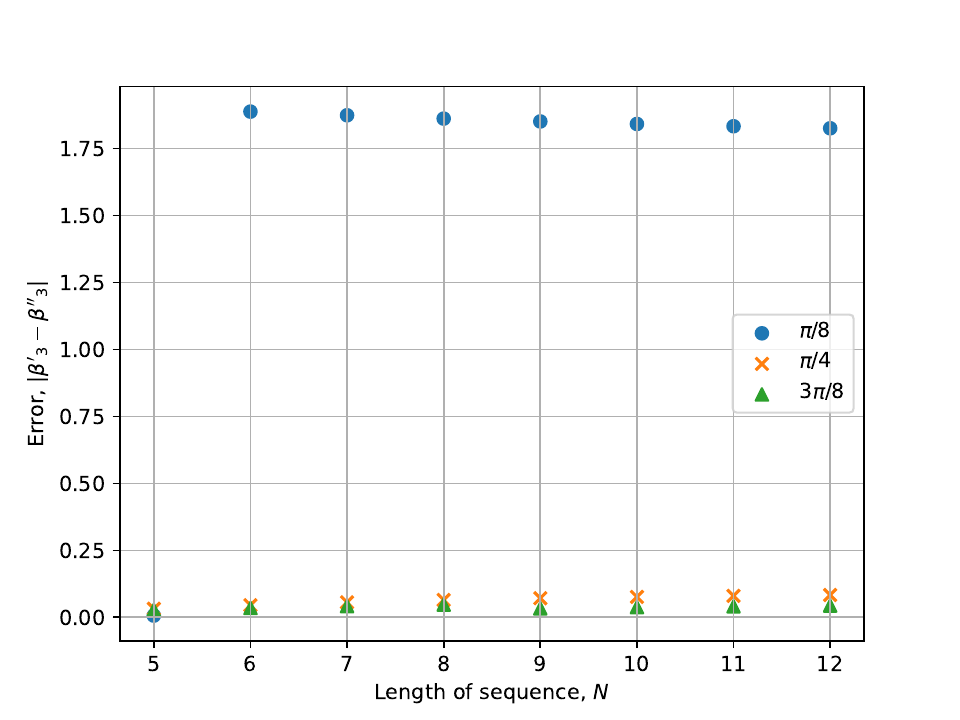}
        \caption{Absolute value of the difference between the solution and the heuristic for $\theta \in \{\pi/8, \pi/4, 3\pi/8\}$ as a function of $N$.}
    \end{subfigure}
    \caption{
    The final case is where there are three stages of mutations: 
    (1) from $\ket{0}$ to \\ $\ket{\phi_1} = \cos(\theta) + \sin(\theta)$, then
    (2) from $\ket{\phi_1} = \cos(\theta) + \sin(\theta)$ to $\ket{\phi_2} = \cos(2\theta) + \sin(2\theta)$, and finally
    (3) from $\ket{\phi_2} = \cos(2\theta) + \sin(2\theta)$ to $\ket{\phi_3} = \cos(3\theta) + \sin(3\theta)$, for $\theta \in \{\pi/8, \pi/4, 3\pi/8\}$.
    We contrast the optimal reward function values for varying $\theta$ as a function of the length of the sequence $N$ (on the left), and the absolute difference between the heuristic and the solution (on the right).
    Observe that when $\theta = \pi/8$, the SDP solver is unable to compute the solution and hence Eq.~\eqref{reduceddual} is assigned of $-1$ for sequences of length larger than $5$.
    More importantly, the heuristic manages to obtain a solution.
    This explains the large error for $\theta = \pi/8$ (on the right) which is more of an indication of the unsolvability than it is an actual error.
    } 
    \label{fig:3CP}
\end{figure}

\end{document}